\documentclass[journal]{IEEEtran}
\usepackage{amsmath}
\usepackage{amssymb}
\usepackage{multirow}
\usepackage{amsfonts}
\usepackage{mathrsfs}
\usepackage{amsthm}
\usepackage{latexsym}
\usepackage{graphicx}
\usepackage{subfig}
\usepackage{color,varioref}
\usepackage{longtable}
\usepackage{enumerate}
\usepackage{color}
\usepackage{float}
\newtheorem{remark}{Remark}
\usepackage{caption}
\newtheorem{theorem}{Theorem}[section]

\begin{document}

\title{A Dual Alternating Direction Method of Multipliers for Image Decomposition and  Restoration}

\author{Qingsong Wang,
        ~Chengjing Wang,
        ~Peipei Tang,
        ~and Dunbiao Niu
\thanks{The work of Peipei Tang was supported by the Natural Science Foundation of Zhejiang Province of China under Grant LY19A010028 and the Science $\&$ Technology Development Project of Hangzhou, China under Grant 20170533B22.}
\thanks{Qingsong Wang, School of Mathematics, Southwest Jiaotong University, No.999, Xian Road, West Park, High-tech Zone, Chengdu 611756, China. ({\tt nothing2wang@hotmail.com}).}
\thanks{Chengjing Wang, School of Mathematics, Southwest Jiaotong University, No.999, Xian Road, West Park, High-tech Zone, Chengdu 611756, China. ({\tt renascencewang@hotmail.com}).}
\thanks{Peipei Tang, School of Computing Science, Zhejiang University City College, Hangzhou 310015, China. ({\tt tangpp@zucc.edu.cn}).}
\thanks{Dunbiao Niu, College of Mathematics, Sichuan University, No.24 South Section 1, Yihuan Road, Chengdu 610065,  China. ({\tt dunbiaoniu\_sc@163.com}).}
}

\maketitle

\begin{abstract}
In this paper, we develop a dual alternating direction method of multipliers (ADMM) for an image decomposition model. In this model, an image is divided into two meaningful components, i.e., a cartoon part and a texture part. The optimization algorithm that we develop not only gives the cartoon part and the texture part of an image but also gives the restored image (cartoon part + texture part). We also present the global convergence and the local linear convergence rate for the algorithm under some mild conditions. Numerical experiments demonstrate the efficiency and robustness of the dual ADMM (dADMM). Furthermore, we can obtain relatively higher signal-to-noise ratio (SNR) comparing to other algorithms. It shows that the choice of the algorithm is also important even for the same model.
\end{abstract}


\IEEEpeerreviewmaketitle

\section{Introduction}

\IEEEPARstart{I}{n} computer science, digital image processing is the use of computer algorithms to perform image processing on images. As we know, image decomposition and restoration are the basic methods in image and computer vision science. Recently, they have been playing more and more important roles in digit image processing, machine learning, pattern recognition, and biomedical engineering, etc. As for image decomposition, an image can be divided into two parts, i.e.,  a cartoon part and a texture part. As for image restoration, it includes many types, such as image denoising, image inpainting, image deblurring, and so on. The aim of image decomposition and restoration is to enhance the quality of an image that is degraded by some reason.

The main tasks of this paper are extracting the cartoon and texture components from a degraded image (a blurry and/or missing pixels image) and getting a restored image from the decomposed results. Given an intensity function $f$ which is an image, the image decomposition is to derive $f = u+v$, where $u$ and $v$ represent the cartoon and texture component of the image $f$, respectively. In general, the cartoon part $u$ is an image formed by homogeneous regions and with sharp boundaries, the texture part $v$ is noise or small scale repeated details. Fig.~\ref{clean_decomposition_jpg} illustrates a image
example (the \emph{Weave} image in Fig.~\ref{test_jpg}) of cartoon and texture where the distinctions of between the two components can be visually discernable.
\begin{figure}[htbp]
	\centering  
	\normalsize
	\subfloat[\emph{Weave}]{
		\label{fig:subfig_a}
		\begin{minipage}[t]{0.15\textwidth}
			\centering
			\includegraphics[angle=0,width=1\textwidth]{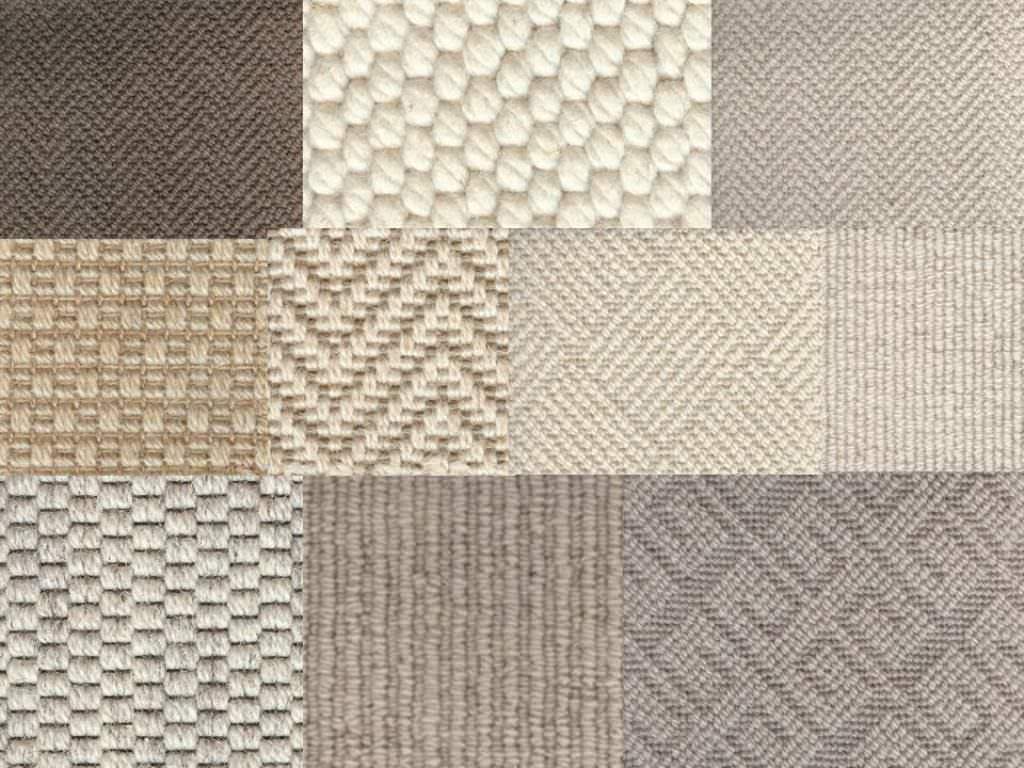}
		\end{minipage}
	}
	\subfloat[cartoon part]{
		\label{fig:subfig_b}
		\begin{minipage}[t]{0.15\textwidth}
			\centering
			\includegraphics[angle=0,width=1\textwidth]{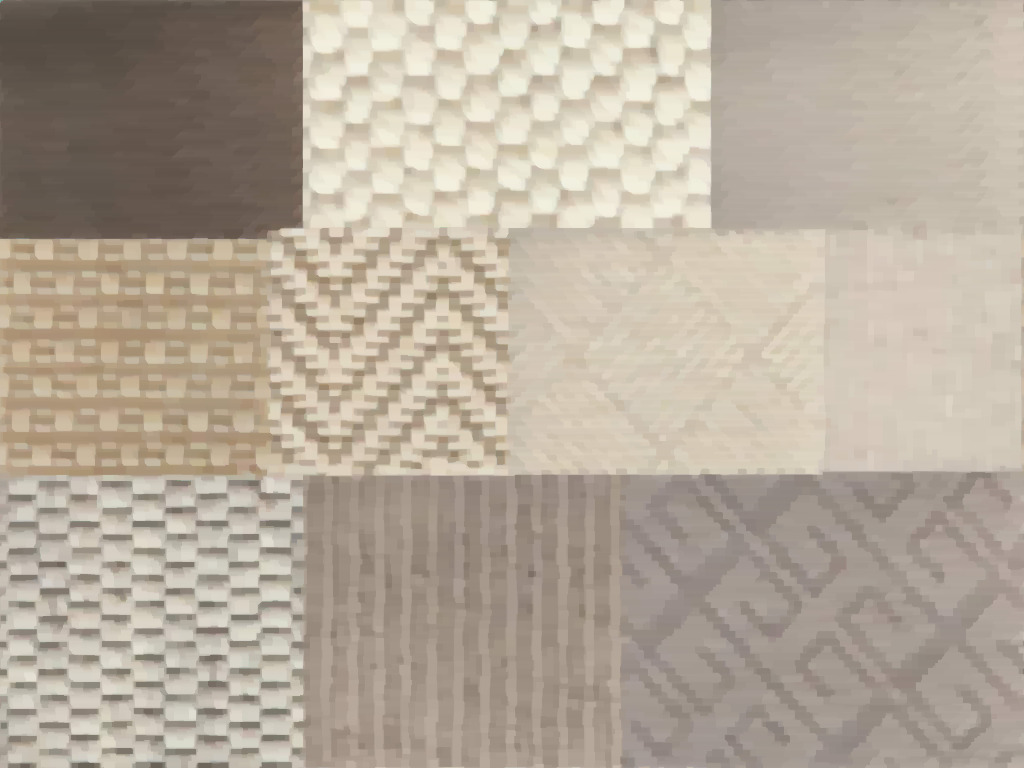}
		\end{minipage}
	}
	\subfloat[texture part]{
		\label{fig:subfig_c}
		\begin{minipage}[t]{0.15\textwidth}
			\centering
			\includegraphics[angle=0,width=1\textwidth]{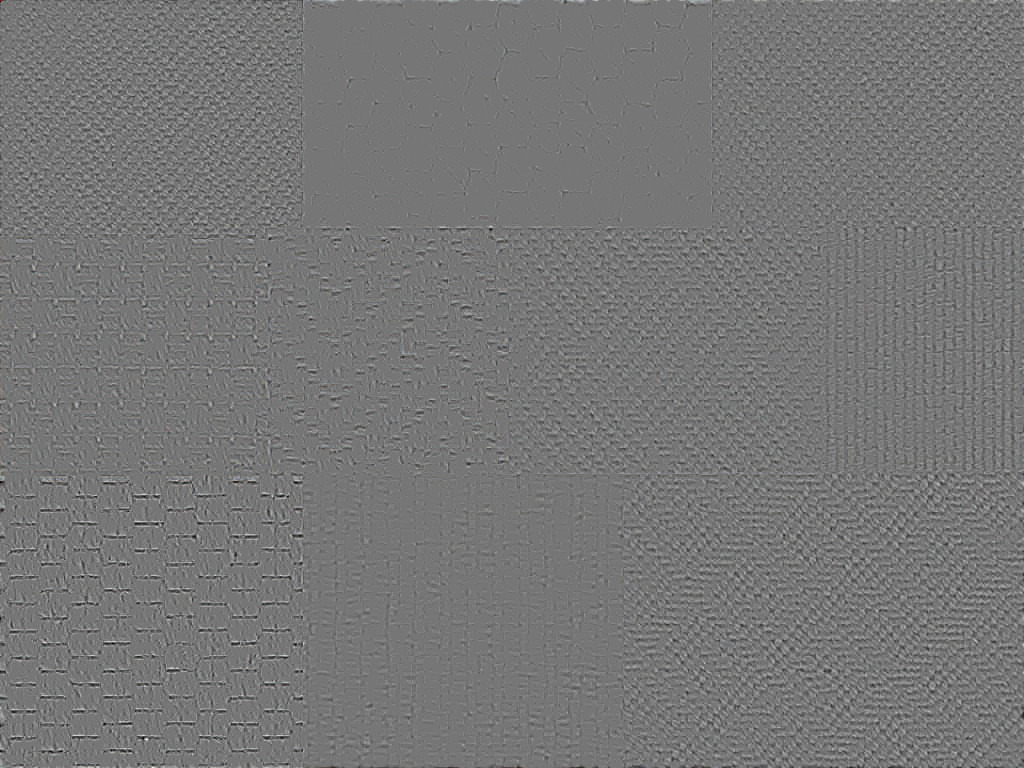}
		\end{minipage}
	}
	\caption{Image decomposition on \emph{Weave} image (\emph{c}) in Fig.~\ref{test_jpg}.}
	\label{clean_decomposition_jpg}
\end{figure}

According to some literatures, the first model about image decomposition is the classic image denoising model proposed in \cite{Rudin_Osher_Fatemi}, i.e., the well-known Rudin, Osher, and Fatemi (ROF) total variation (TV) minimization model for image denoising. The model is given by
\begin{eqnarray}\label{model1}
	\inf_{u}\int_{\Omega}|\nabla u| + \lambda\int_{\Omega}|f-u|^{2}, \label{eq:ROF_original}
\end{eqnarray}
where the set $\Omega$ is a domain of $\textbf{R}^{n}$, $\int_{\Omega}|\nabla u|$ denotes the TV of $u$ with the assumption that $u$ is of bounded variation: $u \in BV(\Omega)$. $\int_{\Omega}|f-u|^{2}$ is a fidelity term and $\lambda>0$ is a weight parameter. Although the problem \eqref{model1} is convex, it is still difficult to solve the functional programming problem. Then we write out the discrete version of the problem \eqref{eq:ROF_original} as follows
\begin{eqnarray}
\begin{aligned}
\min_{u\in\textbf{R}^{n},v\in\textbf{R}^{n}} \quad &\tau\||\nabla u|\|_{1} + \|v\|_{2}^{2} \\
\mbox{s.t.}\,\,\quad\quad & u+v=b,
\end{aligned} \label{eq:ROF_model}
\end{eqnarray}
where $b\in\textbf{R}^{n}$ is a given image, $\nabla :=\begin{pmatrix}\nabla _1\\ \nabla _2\end{pmatrix} : \textbf{R}^{n}\rightarrow \textbf{R}^{n}\times \textbf{R}^{n}$ denotes the discrete version of the first-order derivative operator (see \cite{Adams} for more details), and $\tau>0$ is a control parameter for the decomposition of $b$ into the cartoon part $u$ and the texture part $v$, $\||\nabla u|\|_{1}$ is the TV semi-norm for inducing the cartoon part $u\in \textbf{R}^{n}$, which is defined as

\[
\||\nabla u|\|_{1}:=\sum _{i=1}^n\left[ (\nabla _1u)^2_i+(\nabla _2u)^2_i\right] ^{1/2}.
\]
It has been pointed out in the seminal book \cite{Meyer} that the $l_{1}$ norm or $l_{2}$ norm does not characterize the oscillatory components. This means that these components do not have small norms (see e.g. \cite{Vaillo_Caselles_Mazon}).  In \cite{Meyer}, a new model was proposed by replacing $\|\cdot\|_{2}^{2}$ of the ROF model \eqref{eq:ROF_model} with a weaker norm. The validity of this model has been verified by many numerical experiments in some papers, e.g., \cite{Vese_Osher}, \cite{Aujol_Chambolle}, and \cite{Maure_Aujol_Peyre}. Meyer's model can be equivalently written as
\begin{eqnarray}
\begin{aligned}
\min_{u\in\textbf{R}^{n},v\in\textbf{R}^{n}}\quad &\tau \||\nabla u|\|_{1} + \|v\|_{-1,\infty} \\
\mbox{s.t.}\,\,\quad\quad& u+v = b,
\end{aligned} \label{eq:-1,infinity}
\end{eqnarray}
where $\|v\|_{-1,\infty}$ is the negative semi-norm in Sobolev space for inducing the texture $v$. For any $s\in\left[1,+\infty\right]$,
\begin{equation}
||v||_{-1,s}:=\inf \{ \||g|\|_{s}\,|\,v=\mbox{div}g,\,\, g\in\textbf{R}^{n}\times\textbf{R}^{n} \}, \label{eq:-1,p}
\end{equation}
where $\||g|\|_s = (\sum_{i=1}^{n}|g|_{i}^{s})^\frac{1}{s}$, and $\mbox{div}:=-\nabla^{T} :  \textbf{R}^{n}\times \textbf{R}^{n} \rightarrow  \textbf{R}^{n} $ is the divergence operator. For the details about the operator $\mbox{div}$, one may refer to \cite{Adams}.

It is very difficult to get a numerical solution because the Euler-Lagrange equation of \eqref{eq:-1,infinity} can not be written out directly. There are many image decomposition models which have overcome this difficulty. For example, motivated by the approximation to the $l_\infty$ norm of  $|g|_{i}=\sqrt{(g_{1})^{2}_{i}+(g_{2})^{2}_{i}}$, $g=(g_1,g_2)\in \textbf{R}^{n}\times\textbf{R}^{n}$, and
\[
\||g|\|_{\infty}=\|\sqrt{g_{1}^{2}+g_{2}^{2}}\|_{\infty} = \lim_{s\rightarrow\infty}\|\sqrt{g_{1}^{2}+g_{2}^{2}}\|_{s},
\]
the following convex minimization problem is used in \cite{Vese_Osher}.
\begin{equation}
\min_{u\in\textbf{R}^{n},g\in\textbf{R}^{n}\times\textbf{R}^{n}}\, \tau \||\nabla u|\|_{1} + \frac{1}{2}\|u+\mbox{div} g-b\|_{2}^{2} + \mu\||g|\|_{s}, \label{eq:Euler-Lagrangian}
\end{equation}
where $u$ is the cartoon part of an image, $v=\mbox{div} g$ is the texture part of an image, $s\ge1$, and $\tau>0$, $\mu>0$ are two control parameters. The first and third terms are penalty terms, and the second term satisfies $b\approx u+\mbox{div} g$. This is the first practical method to overcome the difficulty.
In \cite{Vese_Osher}, the gradient method is employed to solve the Euler-Lagrange equation of \eqref{eq:Euler-Lagrangian}. It has also been demonstrated that
the solution of \eqref{eq:Euler-Lagrangian} can be used for texture discrimination and
segmentation.

In \cite{Osher_Sole_Vese}, Osher \textit{et al.} proposed an alternative model by replacing $\|v\|_{-1,\infty}$ in the problem \eqref{eq:-1,infinity} with $\|v\|_{-1,2}$. Then the model is
\begin{eqnarray*}
\begin{aligned}
\min_{u\in\textbf{R}^{n},v\in\textbf{R}^{n}}\quad &\tau \||\nabla u|\|_{1} + \|v\|_{-1,2} \\
\mbox{s.t.}\,\,\quad\quad& u+v = b.
\end{aligned}
\end{eqnarray*}

In \cite{Aujol_Gilboa_Chan_Osher}, Aujol \textit{et al.} studied a constrained optimization model for image image decomposition, that is,
\begin{eqnarray*}
\begin{aligned}
\min_{u\in\textbf{R}^{n},v\in\textbf{R}^{n}}\quad &\||\nabla u|\|_{1} + \frac{1}{2\sigma}\|b-u-v\|^{2}_{2} \\
\mbox{s.t.}\,\,\quad\quad& \|v\|_{-1,\infty} \le \mu,
\end{aligned}
\end{eqnarray*}
where $\sigma$ and $\mu$ are two positive parameters to decompose the image $b$ into the cartoon part $u$ and texture part $v$. Yin \textit{et al.} \cite{Yin_Goldfarb_Osher} studied a second-order cone programming problem, and the problem \eqref{eq:-1,infinity} is a special case of this problem. An overview of image and signal decomposition by using sparsity and morphological diversity was given in \cite{Fadili_Starck_Bobin_Moudden}.

As for image restoration, it is often formulated as an inverse problem. For a degraded image $b$, we recover an unknown clean image $x$  from $b$ such that
\[
b\approx Hx,
\]
where the degradation operation $H:\textbf{R}^{n}\rightarrow\textbf{R}^{n}$ is a linear operator. $Hx$ may be contaminated by some noises, such as the additive noise (for example, Gaussian noise, or impulse noise, i.e., $b=Hx + \varepsilon$), Poisson noise and other multiplicative noise. $H$ can be an identity operator (for denoising), a convolution operator (for deblurring), a projection operator (for inpainting) or the mixing of these operators. When we decompose an image $x$ into two parts, i.e., the cartoon part $u$ and the texture part $v$, it satisfies
\begin{eqnarray*}
 b\approx H(u + v) = H(u + \mbox{div} g).
\end{eqnarray*}
In this paper, we study an image decomposition and reconstruction model for images with degradations. The optimization problem to be considered can be explicitly written as
\begin{equation}
\min_{u\in\textbf{R}^{n},g\in\textbf{R}^{n}\times\textbf{R}^{n}} \tau \||\nabla u|\|_{1} + \frac{1}{2}\|H(u+\mbox{div} g)-b\|_{2}^{2} + \mu\||g|\|_{s}, \label{eq:target}
\end{equation}
where $\tau>0$ and $\mu>0$ are the trade-off parameters between the cartoon part $u$ and the texture part $\mbox{div} g$, respectively;  and $s\ge 1$. Actually only the cases of $s=1$, $s=2$ and $s=\infty$ are considered in the numerical experiments (see Section III for details).

There are some algorithms to solve problem \eqref{eq:target}. The alternating direction method of multipliers (ADMM) \cite{Glowinski_Marrocco} and the ADMM with Gaussian back substitution in \cite{Ng_Yuan_Zhang} are often used methods. Most of the ADMM type methods are applied to the primal problem (pADMM). The proximal gradient (forward-backward splitting) method in \cite{Cai_Chan_Shen} is a special gradient descent method, which is mainly used to solve the optimization problem of non-differentiable objective function. Furthermore, the alternating minimization (AM) based method \cite{Aujol_Aubert} and the partial differential equation (PDE) based method \cite{Vese_Osher} have also been used to solve this kind of problem. Among these algorithms, by and large, the pADMM is relatively most effective and robust. However, for some problems the pADMM is less efficient which also affect to some extent the quality of the image decomposition and restoration.

The rest of this paper is organized as follows. In Section II, we develop a dual ADMM (dADMM) for a more general model problem. Next, some details about the algorithm are explained. We also present the global convergence and the local linear convergence rate of the algorithm, respectively. Some numerical experiments for image decomposition and restoration are presented in Section III. Finally, we give the conclusion in Section IV.

\subsection{Preliminaries}
In this subsection, we summarize some notations of convex optimization which will be used in the subsequent analysis.

For a given proper convex function $f:\textbf{R}^{n}\rightarrow (-\infty,+\infty]$, the proximal mapping $\mbox{Prox}_{\sigma f}(\cdot)$ of $f$ with positive parameter $\sigma$ is defined by
\[
\mbox{Prox}_{\sigma f}(x):= \arg\min_{u\in\textbf{R}^{n}}\,\,\{f(u) + \frac{1}{2\sigma}\|u-x\|_{2}^{2}\},\quad \forall\ x \in \mbox{dom}f,
\]
where $\mbox{dom}f$ is defined as $\mbox{dom}f:=\{x\in\textbf{R}^{n}\,|\,f(x)<+\infty\}$. Next, the Moreau identity is given by
\[
\mbox{Prox}_{\sigma f}(x) + \sigma\mbox{Prox}_{f^{*}/\sigma}(x/\sigma) = x.
\]

For a given function $f:\textbf{R}^{n}\rightarrow [-\infty,+\infty]$, the conjugate of $f$ is defined as
\[
f^{*}(y) = \sup_{x\in\textrm{dom}f}\{ \left<y,x\right> - f(x) \}.
\]
The function $f^{*}$ is closed and convex (even when $f$ is not).

For more details, one may refer to \cite{Rockafellar}.

\section{The Algorithm}
In this section, we first introduce a general optimization problem. Then the dADMM is applied to solve the problem. Based on the convergence result of the semi-proximal ADMM which was summerized in \cite{Fazel_Pong_Sun_Tseng} and the linear convergence rate result which was established in \cite{Han_Sun_Zhang}, we will establish the global convergence and the local linear convergence rate of our algorithm.

Now we consider a general optimization problem as the following form
\begin{equation}
\min_{x\in\mathcal{X},y\in\mathcal{Y}} \quad \frac{1}{2}\|Ax+By-b\|_{2}^{2} + p(x)+q(y), \label{eq:general}
\end{equation}
where $\mathcal{X}$, $\mathcal{Y}$ and $\mathcal{Z}$ are three finite-dimensional real Euclidean spaces each equipped with an inner product $\left<\cdot,\cdot\right>$ and its induced norm $\|\cdot\|$, $b\in \mathcal{Z}$ is a given variable, $A: \mathcal{X}\rightarrow \mathcal{Z}$ and $B: \mathcal{Y}\rightarrow \mathcal{Z}$ are two linear operators, $p(x)$ and $q(y)$ are two closed proper convex functions (which may be nonsmooth).

 It is quite clear that the problem \eqref{eq:target} is a special case of the problem \eqref{eq:general} with the variables $x\in\textbf{R}^{n}$, $y\in\textbf{R}^{n}\times\textbf{R}^{n}$, $b\in\textbf{R}^{n}$, the linear operators $A=H$ and $B=H\cdot\mbox{div}$, the functions $p(\cdot)=\tau \||\nabla (\cdot)|\|_{1}$ and $q(\cdot)=\mu\||\cdot|\|_{s}$.

By introducing a slack variable $z$, the optimization problem \eqref{eq:general} is equivalent to
\[
\begin{aligned}
\min_{x\in\mathcal{X},y\in\mathcal{Y},z\in\mathcal{Z}}\quad &\frac{1}{2}\| z \|_{2}^{2} + p(x) + q(y)\\
\mbox{s.t.}\quad\quad\quad & Ax+By-b=z.
\end{aligned} \tag{P}
\]
The Lagrangian function associated with problem $(P)$ is given by
\[
l(z,x,y;u)=\frac{1}{2}\| z \|_{2}^{2} + p(x) + q(y) + \langle u,Ax+By-z-b \rangle.
\]
Then
\[
\begin{aligned}
&\max_{u\in\mathcal{Z}}\inf_{z,x,y}l(z,x,y;u)\\
&=\max_{u\in\mathcal{Z}}\{\inf_{z}\{\frac{1}{2}\| z \|_{2}^{2} - \left<u,z\right> \} + \inf_{x}\{p(x)+\left<A^{*}u,x\right> \}\\
&\quad + \inf_{y}\{ q(y) +\left<B^{*}u,y\right> \} - \left<u,b\right>\}\\
&=\max_{u\in\mathcal{Z}}\{-\frac{1}{2}\| u \|_{2}^{2} - \sup_{x}\{\left<-A^{*}u,x\right>-p(x)\}\\
&\quad -\sup_{y}\{ \left<-B^{*}u,y\right> - q(y)\} - \left<u,b\right>\}\\
&=\max_{u\in\mathcal{Z}}\{-\frac{1}{2}\| u \|_{2}^{2} - p^{*}(-A^{*}u) - q^{*}(-B^{*}u) - \left<u,b\right>\}.
\end{aligned}
\]
Thus the dual of problem $(P)$ can be explicitly written as
\begin{equation}
\min_{u\in\mathcal{Z}}\quad \frac{1}{2}\| u \|_{2}^{2} + p^{*}(-A^{*}u) + q^{*}(-B^{*}u) + \left<u,b\right>. \label{eq:dual}
\end{equation}
By introducing two slack variables $v$ and $w$, \eqref{eq:dual} can be equivalently rewritten as
\[
\begin{aligned}
\min_{u\in\mathcal{Z},v\in\mathcal{X},w\in\mathcal{Y}}\quad &\frac{1}{2}\|u\|_{2}^{2} + \left<u,b\right> + p^{*}(v) + q^{*}(w)\\
\mbox{s.t.} \quad\quad\quad &-A^{*}u - v=0,\\
&-B^{*}u - w=0 .
\end{aligned} \tag{D}
\]
Then the augmented Lagrangian function associated with problem $(D)$ is given by
\[
\begin{aligned}
L_{\sigma}(u,v,w;x,y)=&\frac{1}{2}\| u \|_{2}^{2} + \left<u,b\right> + p^{*}(v) + q^{*}(w)\\
& + \left<x,-A^{*}u - v\right> +\left<y,-B^{*}u - w\right>\\
&+\frac{\sigma}{2}\| A^{*}u + v \|_{2}^{2} + \frac{\sigma}{2}\| B^{*}u + w \|_{2}^{2}.
\end{aligned}
\]

Now we describe the algorithmic framework of the dADMM as follows.
\begin{center}
	\fbox{\parbox{3.3in}{\textbf{The dADMM for problem $(D)$} :\\
			Given $\varepsilon > 0$. Let $\tau\in(0,(1+\sqrt{5})/2)$ be a scalar parameter, $\sigma > 0$ be a given arbitrary parameter, and choose $u^{0}, v^{0}, w^{0}, x^{0},y^{0}$. For $k=0,1,2,\dots$, perform the following steps in each iteration:\\
			\textbf{Step 1.} Compute the $u,v,w:$
			\[
			\begin{aligned}
			(I_{\mathcal{Z}}+\sigma AA^{*}+\sigma BB^{*})u^{k+1}&=Ax^{k}+By^{k} -b\\
			&\quad\,\,\, -\sigma Av^{k} - \sigma Bw^{k}, \\
			\left[
			\begin{matrix}
			v^{k+1}\\
			w^{k+1}
			\end{matrix}
			\right]
			&=
			\left[
			\begin{matrix}
			\mbox{Prox}_{p^{*}/\sigma}(\frac{x^{k}}{\sigma} - A^{*}u^{k+1} )\\
			\mbox{Prox}_{q^{*}/\sigma}(\frac{y^{k}}{\sigma} - B^{*}u^{k+1} )
			\end{matrix}
			\right],
			\end{aligned}
			\]
			\textbf{Step 2.} Update $x,y:$
			\[
			\begin{aligned}
			x^{k+1} &= x^{k} + \tau\sigma(-A^{*}u^{k+1} - v^{k+1}),\\
			y^{k+1} &= y^{k} + \tau\sigma(-B^{*}u^{k+1} - w^{k+1}),\\
			\end{aligned}
			\]
			\textbf{Step 3.} If a termination criterion is not satisfied, go to \textbf{Step 1}.			
	}}
\end{center}

Then we describe the details of how to solve the $u$-subproblem and $(v,\,w)$-subproblem of the algorithm.
\begin{itemize}
\item $u$- subproblem:
\[
\begin{aligned}
\bar{u}=\arg\min_{u\in\mathcal{Z}}\,\, &\frac{1}{2}\ \| u \|_{2}^{2} + \left<u,b\right> - \left<x,A^{*}u + v\right> \\&-\left<y,B^{*}u + w\right> + \frac{\sigma}{2}\| A^{*}u + v\|_{2}^{2}\\
& + \frac{\sigma}{2}\|B^{*}u + w\|_{2}^{2},
\end{aligned}
\]
which is equivalent to solve the following linear system
\begin{eqnarray}
(I_{\mathcal{Z}}+\sigma AA^{*}+\sigma BB^{*})u &=& Ax+By -b - \sigma Av\nonumber\\
&& - \sigma Bw.
\label{eq:linear system}
\end{eqnarray}
In \eqref{eq:linear system}, $I_{\mathcal{Z}}$ denotes the identity matrix in the space $\mathcal{Z}$. There are similar notations behind and we will not give more explanations.
\item $(v,\,w)$- subproblem:
\[
\begin{aligned}
\bar{v}&=\arg\min_{v\in\mathcal{X}}\,\, \{p^{*}(v)  - \left<x,A^{*}u + v\right> +\frac{\sigma}{2}\|A^{*}u + v\|_{2}^{2} \}\\
&=\arg\min_{v\in\mathcal{X}}\,\,\{p^{*}(v) + \frac{\sigma}{2}\|v-(\frac{x}{\sigma}-A^{*}u)\|_{2}^{2}\}\\
&=\mbox{Prox}_{p^{*}/\sigma}(\frac{x}{\sigma} - A^{*}u ),
\end{aligned}
\]
\[
\begin{aligned}
\bar{w}&=\arg\min_{w\in\mathcal{Y}}\,\, \{q^{*}(w)  - \left<y,B^{*}u + w\right> +\frac{\sigma}{2}\|B^{*}u + w\|_{2}^{2} \}\\
&=\arg\min_{w\in\mathcal{Y}}\,\,\{ q^{*}(w) + \frac{\sigma}{2}\|w-(\frac{y}{\sigma}-B^{*}u)\|_{2}^{2} \}\\
&=\mbox{Prox}_{q^{*}/\sigma}(\frac{y}{\sigma} - B^{*}u ).
\end{aligned}
\]
\end{itemize}
Based on the Moreau identity, we have
\[
\begin{aligned}
\mbox{Prox}_{p^{*}/\sigma}( \frac{x}{\sigma} -A^{*}u)&=\frac{x -\sigma A^{*}u - \mbox{Prox}_{\sigma p}(x -\sigma A^{*}u)}{\sigma},\\
\mbox{Prox}_{q^{*}/\sigma}(\frac{y}{\sigma} - B^{*}u)&=\frac{y -\sigma B^{*}u-\mbox{Prox}_{\sigma q}(y -\sigma B^{*}u)}{\sigma}.
\end{aligned}
\]

Next, we adapt the results developed in  \cite{Fazel_Pong_Sun_Tseng}, \cite{Han_Sun_Zhang} and \cite{Rockafellar} to establish the global convergence in Theorem II.1 and the local linear convergence rate of the dADMM in Theorem II.2, respectively.
\begin{theorem}
	We consider the situation of $\mathcal{S}=0$ and $\mathcal{T}=0$ in (\cite{Fazel_Pong_Sun_Tseng}, Theorem B.1). Assume that the solution set of $(D)$ is nonempty and that the constraint qualification holds. Let the sequence $\{(u^{k},v^{k},w^{k},x^{k},y^{k})\}$ be generated from the dADMM. If $\tau \in (0,(1+\sqrt{5})/2)$, then the sequence $\{ (u^{k},v^{k},w^{k}) \}$ converges to an optimal solution to the problem $(D)$ and the sequence $\{(x^{k},y^{k})\}$ converges to an optimal solution to the problem $(P)$.
\end{theorem}
\begin{proof}
	We can prove the conclusion of this theorem based on Theorem B.1 in \cite{Fazel_Pong_Sun_Tseng} directly. Thus, we omit the details here.
\end{proof}
Before going to the statement of the following theorem, we denote $m:=(u,v,w,x,y)$ with $u\in\mathcal{Z}$, $v,x\in\mathcal{X}$ and $w,y\in\mathcal{Y}$, and $\mathcal{M}:=\mathcal{Z}\times\mathcal{X}\times\mathcal{Y}\times\mathcal{X}\times\mathcal{Y}$. Besides, define the KKT  (Karush-Kuhn-Tucker) mapping \cite{Kuhn_Tucker} $R: \mathcal{M}\rightarrow \mathcal{M}$ as
\[
R(m):=\left[
\begin{matrix}
u+b-Ax-By \\
\left[
\begin{matrix}
v - \mbox{Prox}_{p^{*}}(v + x) \\
w - \mbox{Prox}_{q^{*}}(w + y) \\
\end{matrix}
\right]\\
\left[
\begin{matrix}
-A^{*}u - v \\
-B^{*}u- w \\
\end{matrix}
\right]\\
\end{matrix}
\right],\quad \forall\ m\in \mathcal{M}.
\]
If there exists $(\bar{u},\bar{v},\bar{w},\bar{x},\bar{y})$ which satisfies $R(m)=0$, then $(\bar{u},\bar{v},\bar{w},\bar{x},\bar{y})$ is called a KKT point for the problem $(D)$.

Define a self-adjoint linear operator $\mathcal{P}$ as below
\[
\mathcal{P} := \mbox{Diag}(I_{\mathcal{Z}}, \sigma I_{\mathcal{X}\times\mathcal{Y}}, (\tau\sigma)^{-1}I_{\mathcal{X}\times\mathcal{Y}}) + s_{\tau}\sigma\Gamma\Gamma^{*},
\]
where
\[
s_{\tau} := \frac{5-\tau-3\mbox{min}\{\tau,\tau^{-1}\}}{4},\quad \forall\, \tau\in(0,\infty),
\]
and $\mbox{Diag}(\cdot,\cdot,\cdot)$ is a $3\times 3$ block diagonal linear operator.
Note that $1/4 \le s_{\tau}\le 5/4$, $\forall\, \tau\in(0,(1+\sqrt{5})/2)$,
and let $\Gamma : \mathcal{X}\times \mathcal{Y}\rightarrow \mathcal{M}$ be a linear operator such that its adjoint $\Gamma^{*}$ satisfies
\[
\begin{aligned}
\Gamma^{*}(m) &=
\left[\begin{matrix}
-A^{*}\\
-B^{*}
\end{matrix}
\right]\times
u+ \left[\begin{matrix}
-I_{\mathcal{X}}& 0\\
0 &-I_{\mathcal{Y}}
\end{matrix}
\right]\times
\left[\begin{matrix}
v\\
w
\end{matrix}
\right]\\
&=\left[\begin{matrix}
-A^{*}u-v\\
-B^{*}u-w
\end{matrix}
\right],\quad \forall\, m\in \mathcal{M}.
\end{aligned}
\]
Then we denote $\|m\|_{\mathcal{P}}:=\sqrt{\langle m, \mathcal{P}m \rangle}$ and  $\mbox{dist}_{\mathcal{P}}(m,\Omega):=\inf_{\omega \in \Omega}\|m-\omega\|_{\mathcal{P}} $ for any $m\in \mathcal{M}$ and any set $\Omega \subseteq \mathcal{M}$.
\begin{theorem}
	We consider the situation of $\mathcal{S}=0$ and $\mathcal{T}=0$ in (\cite{Han_Sun_Zhang}, Corollary 1). Let $\tau \in (0,(1+\sqrt{5})/2)$.  Assume that the solution set  $ \bar{\Omega}=\{\bar{m}:=(\bar{u},\bar{v},\bar{w},\bar{x},\bar{y})\}$ is nonempty. Suppose that the mapping $R: \mathcal{M}\rightarrow \mathcal{M}$ is piecewise polyhedral. Then there exist a constant $\hat{\eta}>0$ such that the infinite sequence $\{(u^{k},v^{k},w^{k},x^{k},y^{k})\}$ generated from the dADMM satisfies that for all $k\ge 1$,
	\[
	\begin{aligned}
	\mbox{dist}(m^{k},\bar{\Omega}) &\le \hat{\eta}\|R(m^{k})\|_{2}, \\
	\mbox{dist}_{\mathcal{P}}(m^{k+1},\bar{\Omega})  &\le \hat{\mu}\,\mbox{dist}_{\mathcal{P}}(m^{k},\bar{\Omega}),
	\end{aligned}
	\]
	where $0<\hat{\mu}<1$ is a constant parameter.
\end{theorem}
\begin{proof}
	By using the Corollary 1 in \cite{Han_Sun_Zhang}, the desired conclusions of this theorem can be obtained readily. So we omit the details here.
\end{proof}

\section{Numerical Experiments}
In this section, we perform the numerical experiments for the dADMM to solve the problem $(\mbox{D})$ with different choices of $A$. More specifically, we test four cases in this section:
\begin{itemize}
	\item[(1)] $A=I$, where $I$ is a identity matrix;
	\item[(2)] $A=S$, where $S$ is a blurring matrix;
	\item[(3)] $A=K$, where $K$ is a binary matrix which indicates the missing pixels by zero entries;
	\item[(4)] $A=KS$, where $KS$ is the composition of a binary matrix and a blurring matrix.
\end{itemize}
All of the experiments were implemented in {\sc Matlab} R2018a x64 on a PC with an Intel i5-8600K 3.6 GHz processor and 8GB memory.
In order to measure the quality of the image decomposition, it is assumed that the cartoon part and the texture part of an image are uncorrelated according to the paper \cite{Aujol_Gilboa_Chan_Osher}. The correlation between the cartoon part $u$ and the texture part $v$ is computed by
\[
\mbox{Corr}(u,v) :=\frac{\mbox{cov}(u,v)} {\sqrt{\mbox{var}(u)\mbox{var}(v)}},
\]
where $\mbox{var}(\cdot)$ and $\mbox{cov}(\cdot,\cdot)$ refer to the sample variance and covariance of given data, respectively. We use the PSNR value which is defined by
\[
\begin{aligned}
\mbox{PSNR}&=10\log_{10}\left(\frac{\mbox{I}^{2}_{\max}}{\mbox{MSE}}\right), \\
\end{aligned}
\]
to measure the performance of image restoration,
where $\mbox{I}_{\max}$ is the maximum intensity of the original image, and MSE is defined by
\[
\mbox{MSE}=\frac{1}{mn}\sum_{i=1}^{m}\sum_{j=1}^{n}[\mbox{I}(i,j)-\mbox{N}(i,j)]^{2},
\]
with $\mbox{I}$ being a noise-free $m\times n$ monochrome image and $\mbox{N}$ being the noisy approximation image of \mbox{I}.

\begin{figure*}[htbp]
	\centering   
	\subfloat[\emph{Lena}]{
		\label{fig:subfig_a}
		\begin{minipage}[t]{0.18\textwidth}
			\centering
			\includegraphics[angle=0,width=1\textwidth]{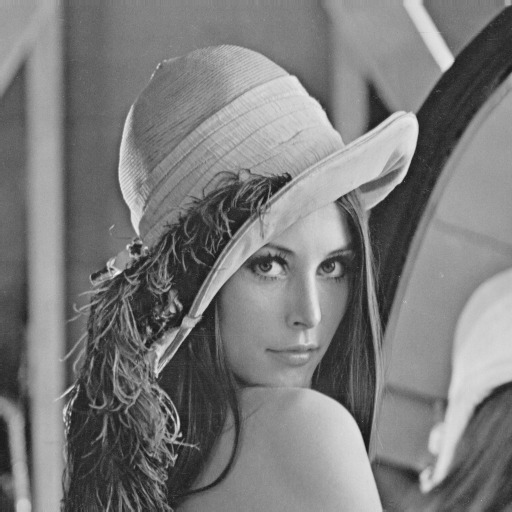}
		\end{minipage}
	}
	\subfloat[\emph{Wool}]{
		\label{fig:subfig_b}
		\begin{minipage}[t]{0.18\textwidth}
			\centering
			\includegraphics[angle=0,width=1\textwidth]{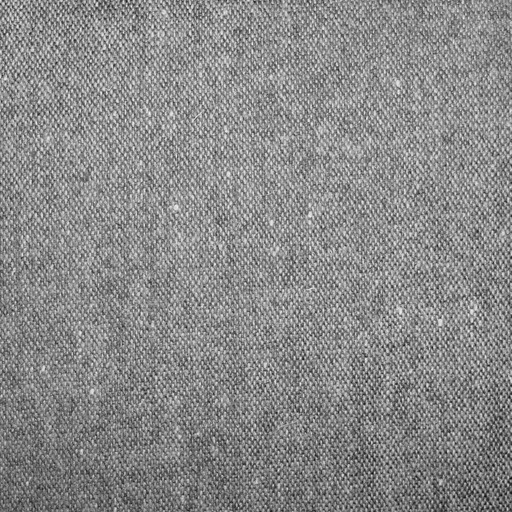}
		\end{minipage}
	}
	\subfloat[\emph{Mixed}]{
		\label{fig:subfig_c}
		\begin{minipage}[t]{0.18\textwidth}
			\centering
			\includegraphics[angle=0,width=1\textwidth]{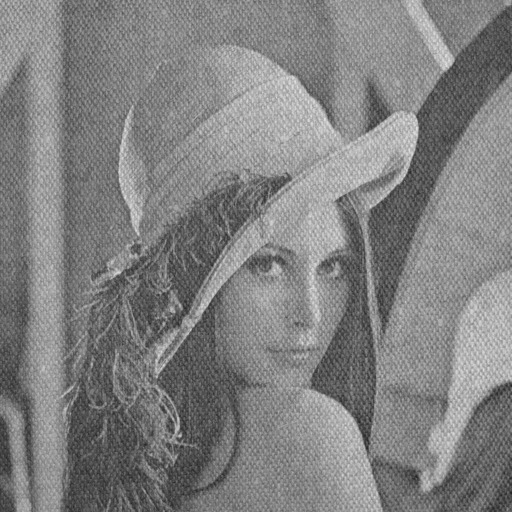}
		\end{minipage}
	}
	\subfloat[\emph{Weave}]{
		\label{fig:subfig_d}
		\begin{minipage}[t]{0.18\textwidth}
			\centering
			\includegraphics[angle=0,width=1\textwidth]{weave_new.jpg}
		\end{minipage}
	}
	\vfill 
	\subfloat[\emph{Barbara\_RGB}]{
		\label{fig:subfig_e}
		\begin{minipage}[t]{0.18\textwidth}
			\centering
			\includegraphics[angle=0,width=1\textwidth]{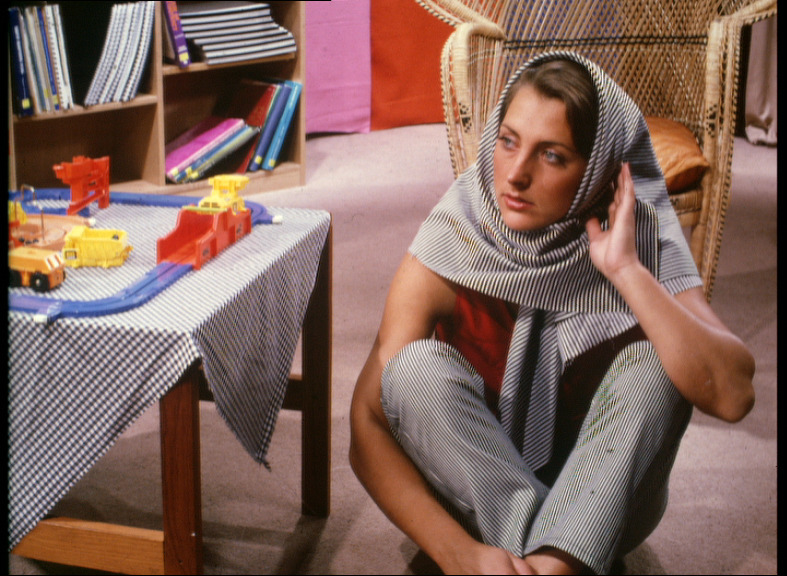}
		\end{minipage}
	}
	\subfloat[\emph{Barbara}]{
		\label{fig:subfig_f}
		\begin{minipage}[t]{0.18\textwidth}
			\centering
			\includegraphics[angle=0,width=1\textwidth]{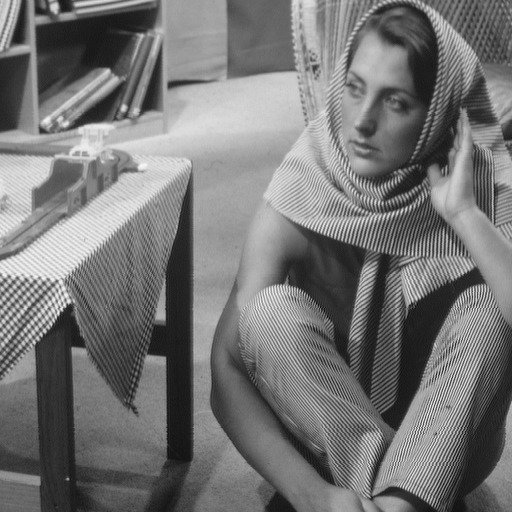}
		\end{minipage}
	}
	\subfloat[\emph{Brick}]{
		\label{fig:subfig_g}
		\begin{minipage}[t]{0.18\textwidth}
			\centering
			\includegraphics[angle=0,width=1\textwidth]{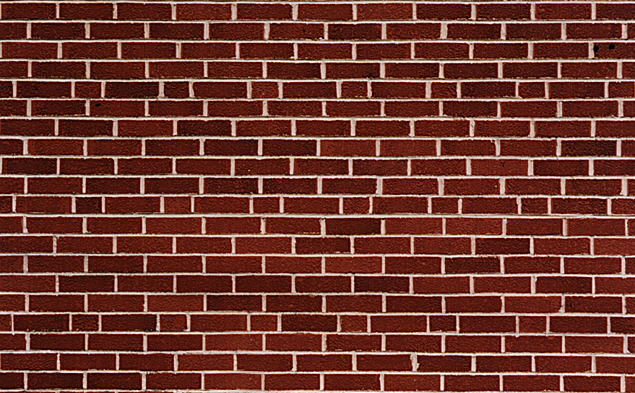}
		\end{minipage}
	}
	\subfloat[\emph{Wood}]{
		\label{fig:subfig_h}
		\begin{minipage}[t]{0.18\textwidth}
			\centering
			\includegraphics[angle=0,width=1\textwidth]{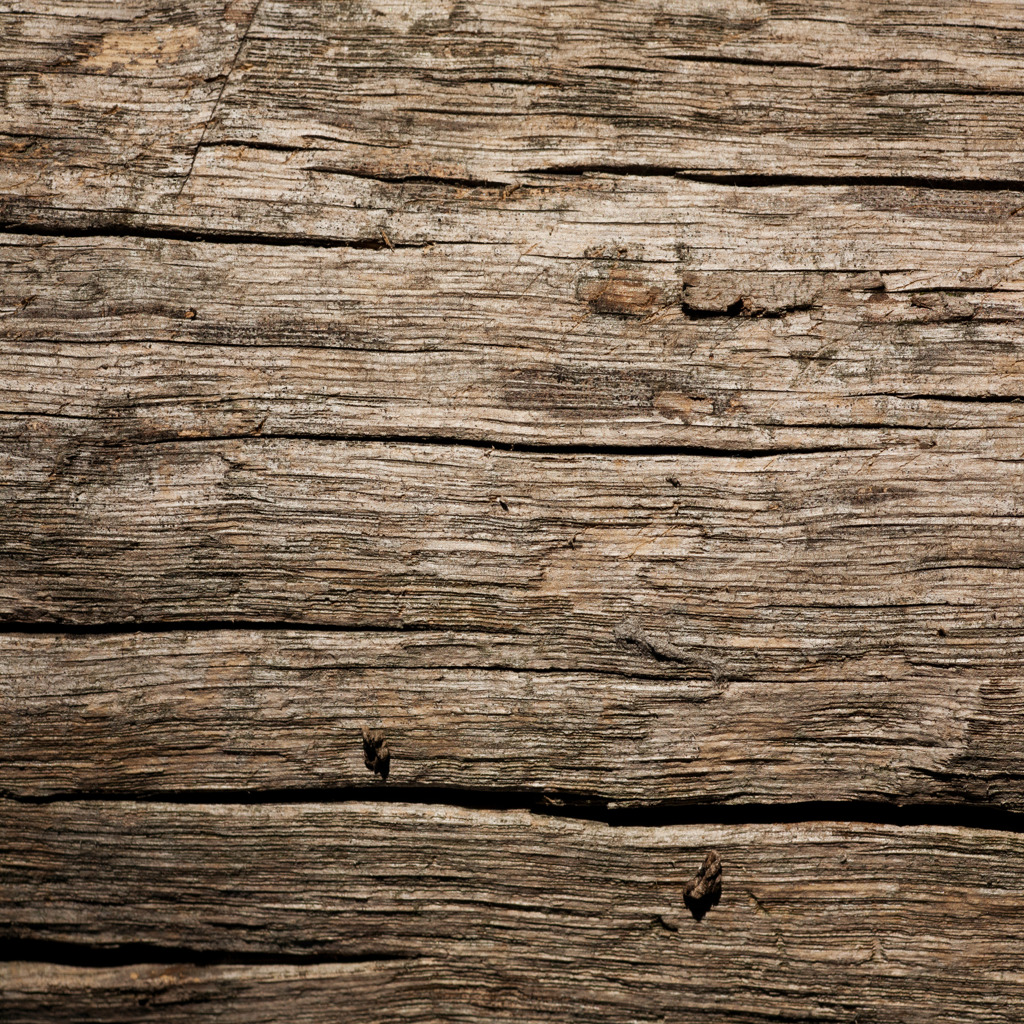}
		\end{minipage}
	}
	\caption{Testing images: (\emph{a}) 512 $\times$ 512 \emph{Lena} image, (\emph{b}) 512 $\times$ 512 \emph{Wool} image, (\emph{c}) combined 512 $\times$ 512 \emph{Lena} and \emph{Wool} image (denoted \emph{Mixed} image),  (\emph{d}) 768 $\times$ 1024 $\times$ 3 \emph{Weave} image, (\emph{e}) 576 $\times$ 787 $\times$ 3 \emph{Barbara\_RGB} image, (\emph{f}) 512 $\times$ 512 a part of \emph{Barbara} image, (\emph{g}) 393 $\times$ 635 $\times$ 3 \emph{Brick} image, (\emph{h})  1024 $\times$ 1024 $\times$ 3 \emph{Wood} image.}
	\label{test_jpg}
\end{figure*}
The images used in this section are displayed in Fig.~\ref{test_jpg}. Note that the \emph{Mixed} image (\emph{c}) in Fig.~\ref{test_jpg} is derived by combining (\emph{a}) and (\emph{b}) with the ratio of 6:4.

We denote
\[
\begin{aligned}
R_{P} &= \frac{\|u+b-Ax-By\|_{2}}{1+\|A\|_{2}}, \\
R_{D} &= \frac{\|A^{T}u + v\|_{2} + \|B^{T}u + w\|_{2}}{1+\|A\|_{2}},\\
R_{C} &= \frac{\|v - \mbox{Prox}_{p^{*}}(v + x)\|_{2} + \|w - \mbox{Prox}_{q^{*}}(w + y)\|_{2}}{1+\|A\|_{2}}.
\end{aligned}
\]
The numerical experiments are terminated if the stopping criterion
\[
\mbox{Tol} = \max\{R_{P},R_{D},R_{C} \}\le 10^{-3}
\]
is satisfied or the maximal iterations reaches 70. This stopping criterion is also used for the other two algorithms (ADME and ADMGB) in \cite{Ng_Yuan_Zhang}.  For all the experiments, the initial iteration point is set to be zero.

\begin{remark}
	Some elements of the texture part $v\in \textbf{R}^{n\times m}$ are less than zero in the experiments. We normalize the matrix $v$ in order to show the texture part of the image clearly. The step of normalization is defined by
	\[
	v_{ij}=\frac{v_{ij}-v^{\mbox{min}}}{v^{\mbox{max}}-v^{\mbox{min}}},\quad \forall\ i=1,\ldots,n,\ j=1,\ldots,m.
	\]
where $v^{\mbox{min}}=\min\limits_{1\leq i\leq n,1\leq j\leq m} \{v_{ij}\}$ and $v^{\mbox{max}}=\max\limits_{1\leq i\leq n,1\leq j\leq m} \{v_{ij}\}$.
\end{remark}
\begin{figure*}[htbp]
	\centering  
	\subfloat{
		\label{fig:subfig_a}
		\begin{minipage}[t]{0.18\textwidth}
			\centering
			\includegraphics[angle=0,width=1\textwidth]{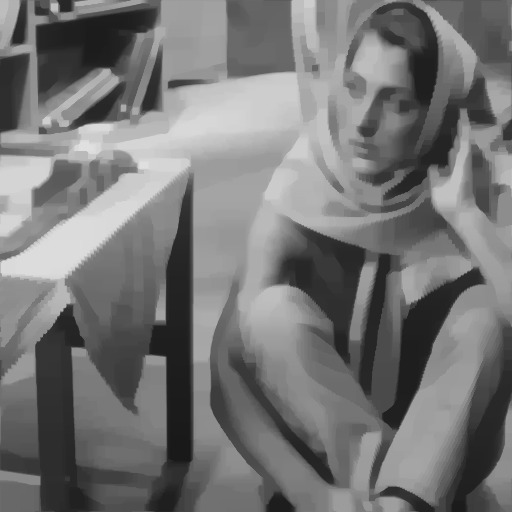}
		\end{minipage}
	}
	\subfloat{
		\label{fig:subfig_b}
		\begin{minipage}[t]{0.18\textwidth}
			\centering
			\includegraphics[angle=0,width=1\textwidth]{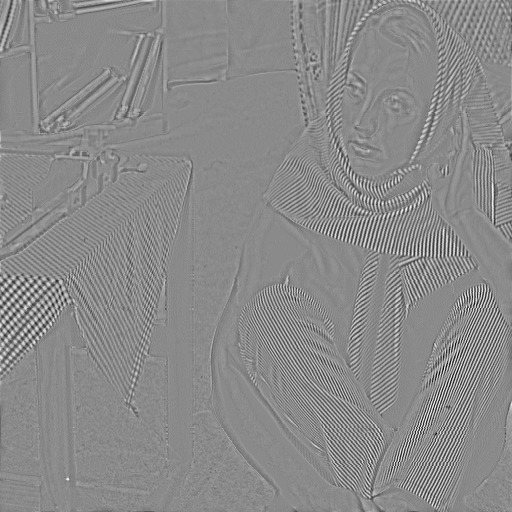}
		\end{minipage}
	}
	\subfloat{
		\label{fig:subfig_c}
		\begin{minipage}[t]{0.18\textwidth}
			\centering
			\includegraphics[angle=0,width=1\textwidth]{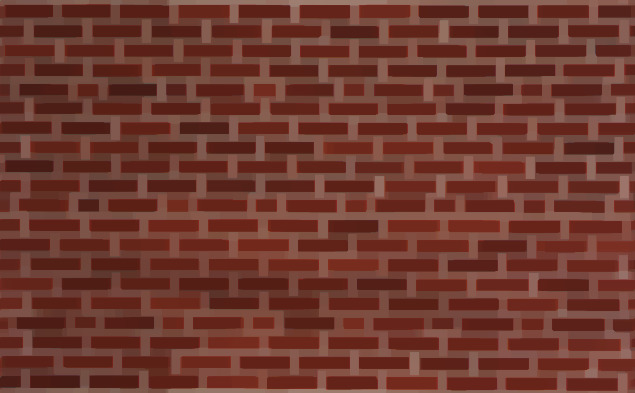}
		\end{minipage}
	}
	\subfloat{
		\label{fig:subfig_d}
		\begin{minipage}[t]{0.18\textwidth}
			\centering
			\includegraphics[angle=0,width=1\textwidth]{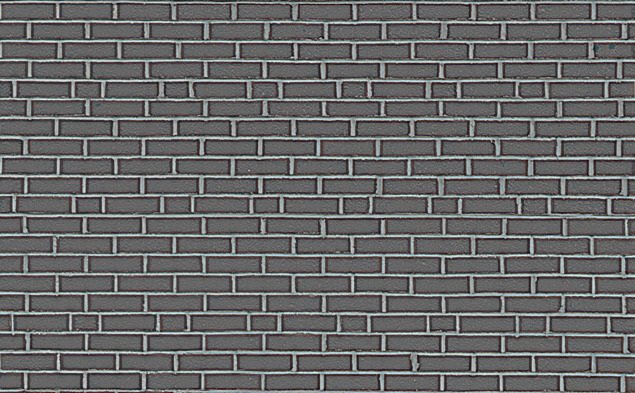}
		\end{minipage}
	}
	\vfill 
	\subfloat{
		\label{fig:subfig_e}
		\begin{minipage}[t]{0.18\textwidth}
			\centering
			\includegraphics[angle=0,width=1\textwidth]{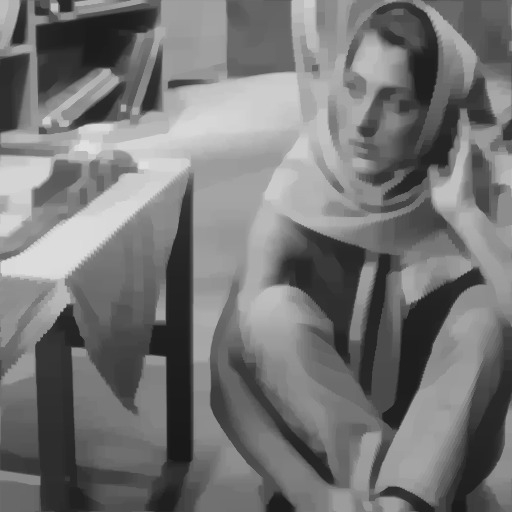}
		\end{minipage}
	}
	\subfloat{
		\label{fig:subfig_f}
		\begin{minipage}[t]{0.18\textwidth}
			\centering
			\includegraphics[angle=0,width=1\textwidth]{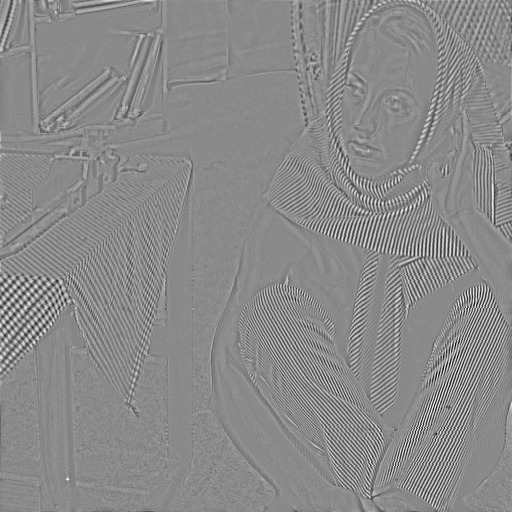}
		\end{minipage}
	}
	\subfloat{
		\label{fig:subfig_g}
		\begin{minipage}[t]{0.18\textwidth}
			\centering
			\includegraphics[angle=0,width=1\textwidth]{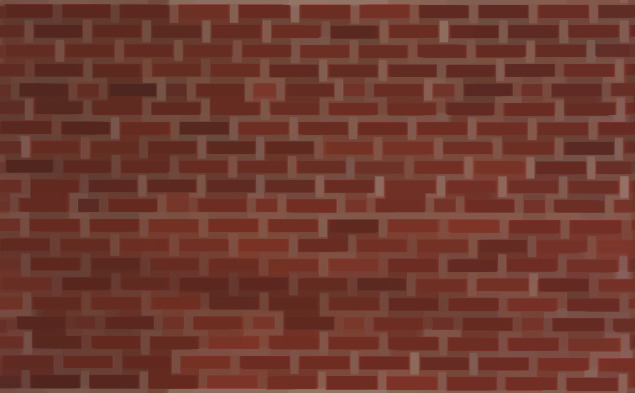}
		\end{minipage}
	}
	\subfloat{
		\label{fig:subfig_h}
		\begin{minipage}[t]{0.18\textwidth}
			\centering
			\includegraphics[angle=0,width=1\textwidth]{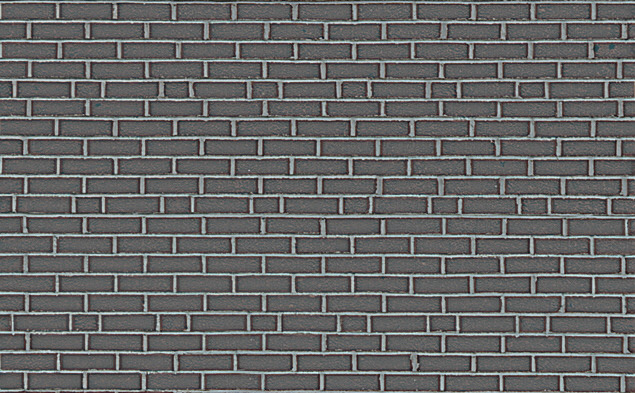}
		\end{minipage}
	}
	
	\vfill 
	
	\subfloat{
		\label{fig:subfig_e}
		\begin{minipage}[t]{0.18\textwidth}
			\centering
			\includegraphics[angle=0,width=1\textwidth]{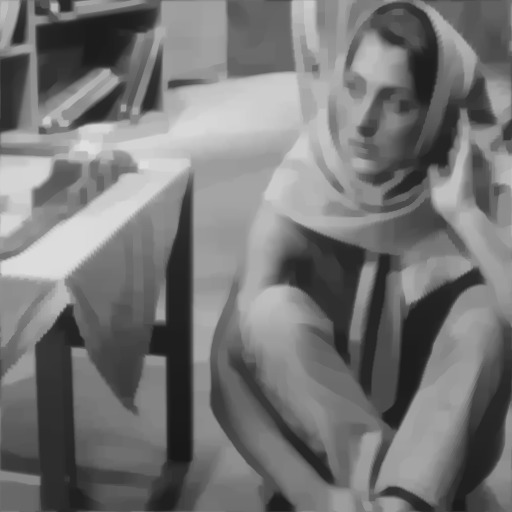}
		\end{minipage}
	}
	\subfloat{
		\label{fig:subfig_f}
		\begin{minipage}[t]{0.18\textwidth}
			\centering
			\includegraphics[angle=0,width=1\textwidth]{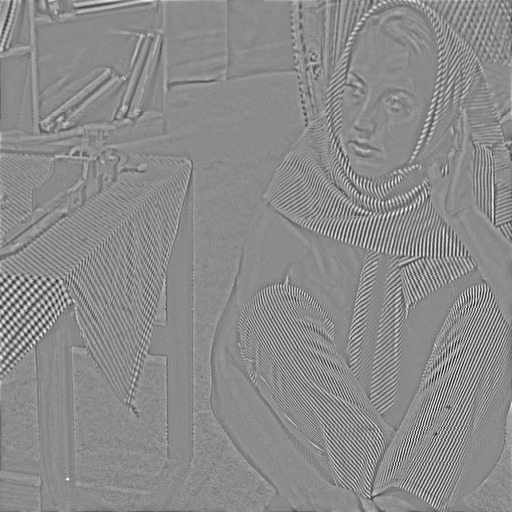}
		\end{minipage}
	}
	\subfloat{
		\label{fig:subfig_g}
		\begin{minipage}[t]{0.18\textwidth}
			\centering
			\includegraphics[angle=0,width=1\textwidth]{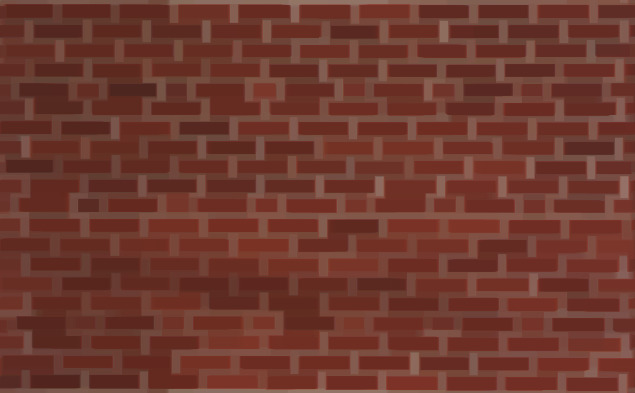}
		\end{minipage}
	}
	\subfloat{
		\label{fig:subfig_h}
		\begin{minipage}[t]{0.18\textwidth}
			\centering
			\includegraphics[angle=0,width=1\textwidth]{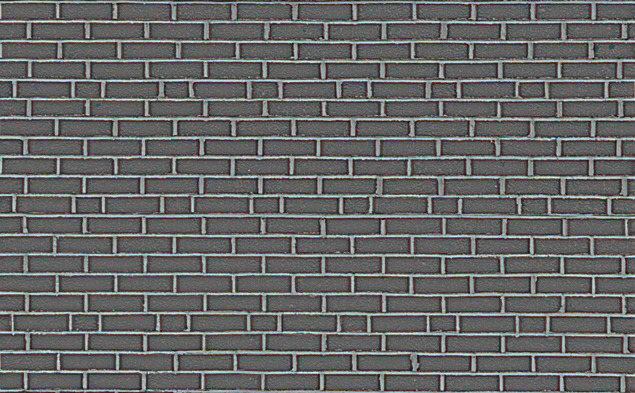}
		\end{minipage}
	}
	\caption{Image decomposition on clean images ((\emph{f}) and (\emph{g}) in Fig.~\ref{test_jpg}, respectively) with different values of $s$. From left to right: the cartoon part, the texture part of (\emph{f}) and (\emph{g}), respectively. The top row: $s=1$. The center row: $s=2$. The bottom row: $s=\infty$.}
	\label{p1_p2_pinf_jpg}
\end{figure*}

\subsection{Case 1: $A=I$}
We consider the case of $A=I$ in this subsection, i.e., a image will be decomposed into a clean and an additive noise images. For the images (\emph{f}) and (\emph{g}) in Fig.~\ref{test_jpg}, we implement the ADME, ADMGB and dADMM in this case. For the ADME and ADMGB, we take the  parameters as $\tau = 1\times10^{-1}$, $\mu =1\times 10^{-3}$ and $\beta_{1} = \beta_{2} =\beta_{3} = 10$. For the dADMM, we set the parameters as $\tau =1 \times 10^{-1}$, $\mu = 3\times 10^{-2}$ and $\sigma= 0.8$.

\begin{figure}[htbp]
	\normalsize
	\subfloat{
		\label{fig:subfig_a}
		\begin{minipage}[t]{0.45\textwidth}
			\centering
			\includegraphics[angle=0,width=1\textwidth]{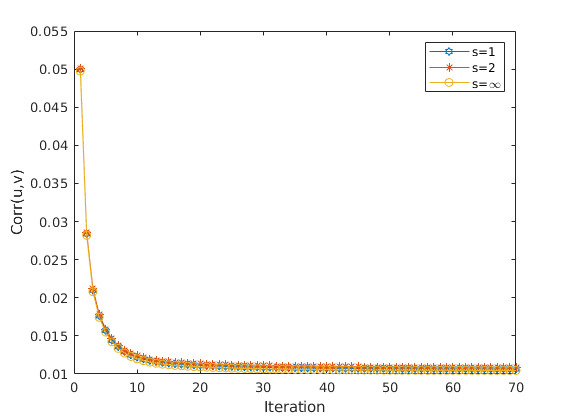}
		\end{minipage}
	}
	\caption{Variations of Corr$(u,v)$ with respect to the iterations for the image (\emph{f}) \emph{Barbara} in Fig.~\ref{test_jpg} with $s=1$, $s=2$, and $s=\infty$.}
	\label{corr_barbara_KI_jpg}
\end{figure}

In Fig.~\ref{p1_p2_pinf_jpg}, the decomposed results (the cartoon part and the texture part) of the images (\emph{f}) and (\emph{g}) with different values of $s$ are displayed. We can hardly see the difference with different values of $s$ from Fig.~\ref{p1_p2_pinf_jpg} directly. In order to measure the quality of the image decomposition, the correlation between the cartoon and the texture is computed. The corresponding results are displayed in Fig.~\ref{corr_barbara_KI_jpg}. The results in Fig.~\ref{corr_barbara_KI_jpg} show a tiny difference of different $s$, which coincides with the conclusion from \cite{Vese_Osher}. These results show that the effectiveness of the optimization problem $(5)$ is not sensitive to different values of $s$. So in the remaining experiments, one of $s=1$, $s=2$ and $s=\infty$ is used to  the numerical simulations for convenience.

Next, we compare the dADMM with the other two algorithms (the ADME and ADMGB in \cite{Ng_Yuan_Zhang}) on image decomposition. For this comparison, we focus on the images (\emph{c}) and (\emph{g}) in Fig.~\ref{test_jpg}. The decomposed results about the image (\emph{c}) in Fig.~\ref{test_jpg} are showed in Fig.~\ref{ADME_ADMGB_ADMMD_jpg} and the variations of the correlation values of the images (\emph{c}) and (\emph{g}) are displayed in Fig.~\ref{corr_clean_KI_jpg}.
\begin{figure}[htbp]
	\centering
	\normalsize
	\subfloat{
		\label{fig:subfig_a}
		\begin{minipage}[t]{0.2\textwidth}
			\centering
			\includegraphics[angle=0,width=1\textwidth]{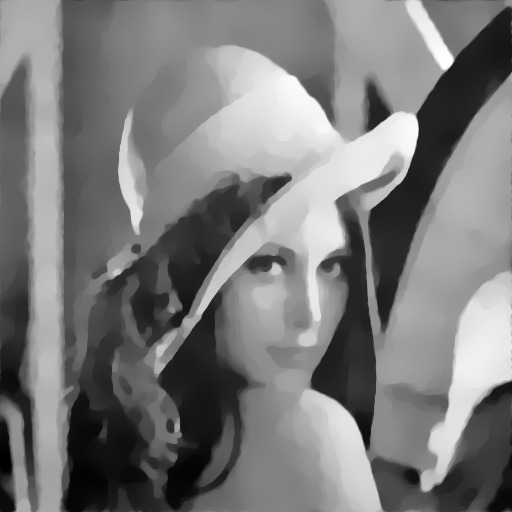}
		\end{minipage}
	}
	\subfloat{
		\label{fig:subfig_b}
		\begin{minipage}[t]{0.2\textwidth}
			\centering
			\includegraphics[angle=0,width=1\textwidth]{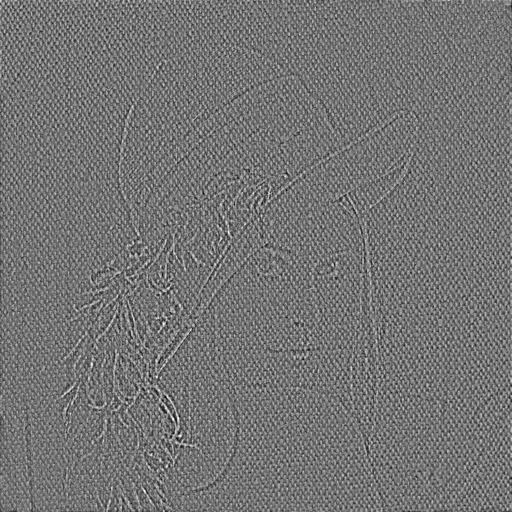}
		\end{minipage}
	}
	\vfill
	\subfloat{
		\label{fig:subfig_c}
		\begin{minipage}[t]{0.2\textwidth}
			\centering
			\includegraphics[angle=0,width=1\textwidth]{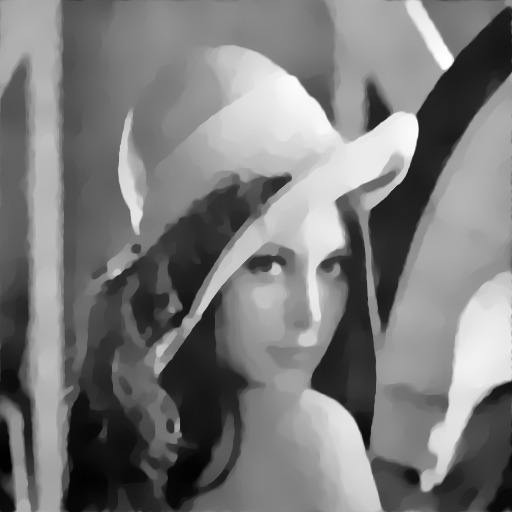}
		\end{minipage}
	}
	\subfloat{
		\label{fig:subfig_d}
		\begin{minipage}[t]{0.2\textwidth}
			\centering
			\includegraphics[angle=0,width=1\textwidth]{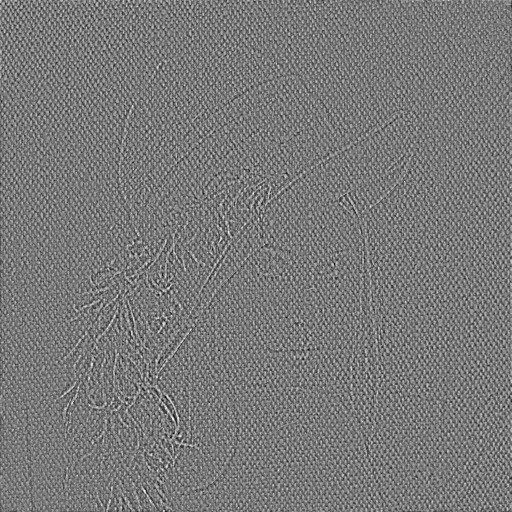}
		\end{minipage}
	}
	\vfill
	\subfloat{
		\label{fig:subfig_e}
		\begin{minipage}[t]{0.2\textwidth}
			\centering
			\includegraphics[angle=0,width=1\textwidth]{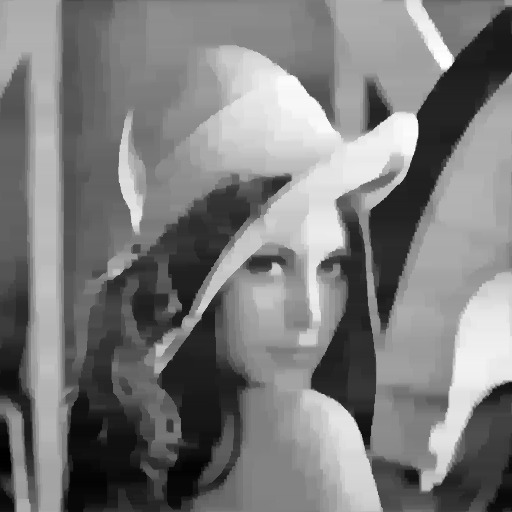}
		\end{minipage}
	}
	\subfloat{
		\label{fig:subfig_f}
		\begin{minipage}[t]{0.2\textwidth}
			\centering
			\includegraphics[angle=0,width=1\textwidth]{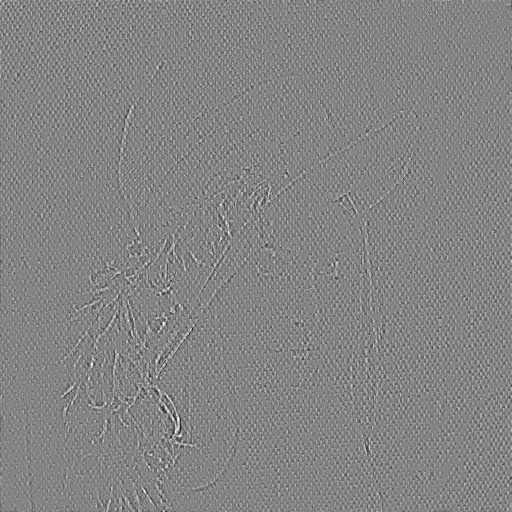}
		\end{minipage}
	}
	\caption{Image decomposition on the \emph{Mixed} image (\emph{c}) in Fig.~\ref{test_jpg} with $A=I$. The first column: the cartoon part. The second column: the texture part. From top to bottom are the decomposed results by the ADME, ADMGB, and dADMM.}
	\label{ADME_ADMGB_ADMMD_jpg}
\end{figure}
\begin{figure}[htbp]
	\normalsize
	\subfloat{
		\label{fig:subfig_a}
		\begin{minipage}[t]{0.45\textwidth}
			\centering
			\includegraphics[angle=0,width=1\textwidth]{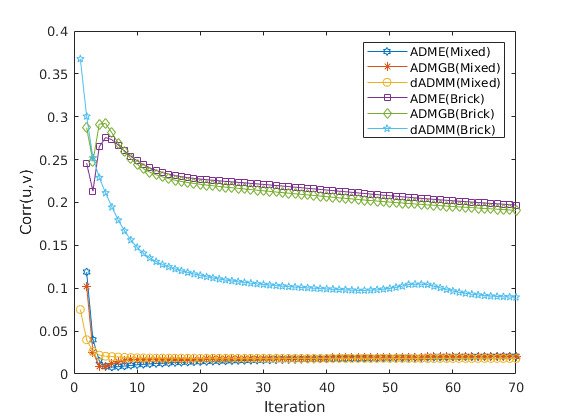}
		\end{minipage}
	}
	\caption{Variations of Corr$(u,v)$ with respect to the iterations for  \emph{Mixed} image (\emph{c}) and  \emph{Brick} image (\emph{g}) in Fig.~\ref{test_jpg}.}
	\label{corr_clean_KI_jpg}
\end{figure}

Now we add some additive noise to the image ($e$) in Fig.~\ref{test_jpg}, which is $b=(u+v)+\varepsilon$, where the additive noise $\varepsilon$ is Gaussian noise. The noised image $b$ is decomposed by the model (5) with three algorithms. In this experiment, the Gaussian white noise $\varepsilon$ is generated by {\sc Matlab} function \texttt{imnoise(Im,`gaussian',0,0.1)}. The results are displayed in Fig.~\ref{noised_ADME_ADMGB_ADMMD_jpg} and Fig.~\ref{corr_noised_KI_jpg}, respectively.
From Fig.~\ref{noised_ADME_ADMGB_ADMMD_jpg}, we note that we can obtain a better cartoon part and a more clear texture part by the dADMM than the other two algorithms (ADME and ADMGB). We plot the variations of $\textrm{Corr}(u,v)$ with respect to the iterations in Fig.~\ref{corr_noised_KI_jpg}. It shows that dADMM can reach a relatively lower correlation value with fewer iterations than the other two algorithms.
\begin{figure}[htbp]
	\centering
	\subfloat{
		\label{fig:subfig_a}
		\begin{minipage}[t]{0.15\textwidth}
			\centering
			\includegraphics[angle=0,width=1\textwidth]{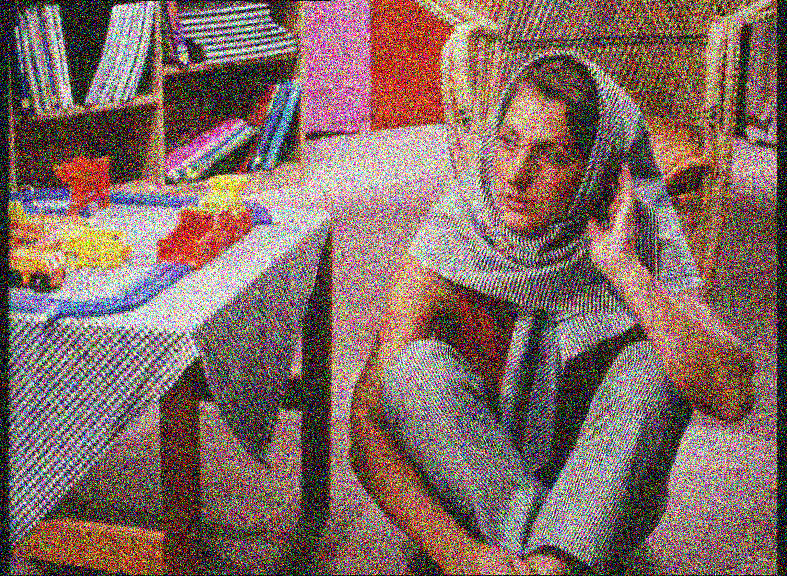}
		\end{minipage}
	}
\vfill
	\subfloat{
		\label{fig:subfig_b}
		\begin{minipage}[t]{0.15\textwidth}
			\centering
			\includegraphics[angle=0,width=1\textwidth]{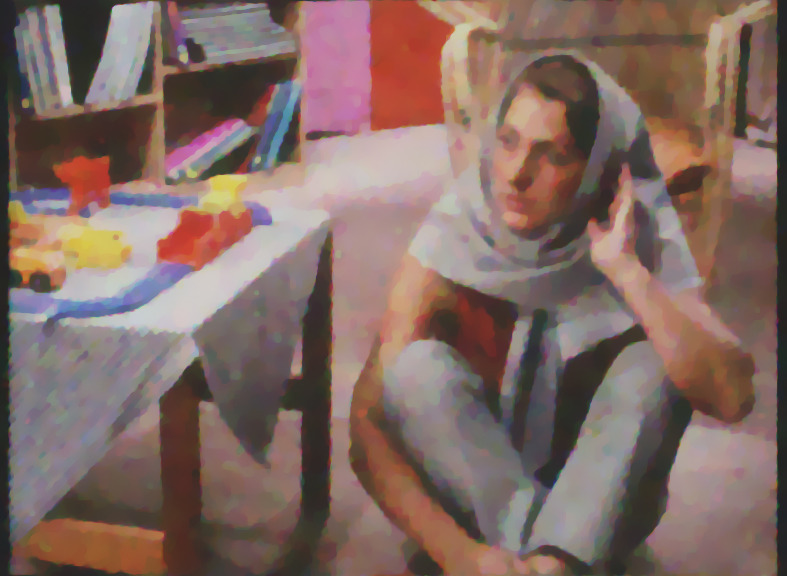}
		\end{minipage}
	}
	\subfloat{
	\label{fig:subfig_f}
	\begin{minipage}[t]{0.15\textwidth}
		\centering
		\includegraphics[angle=0,width=1\textwidth]{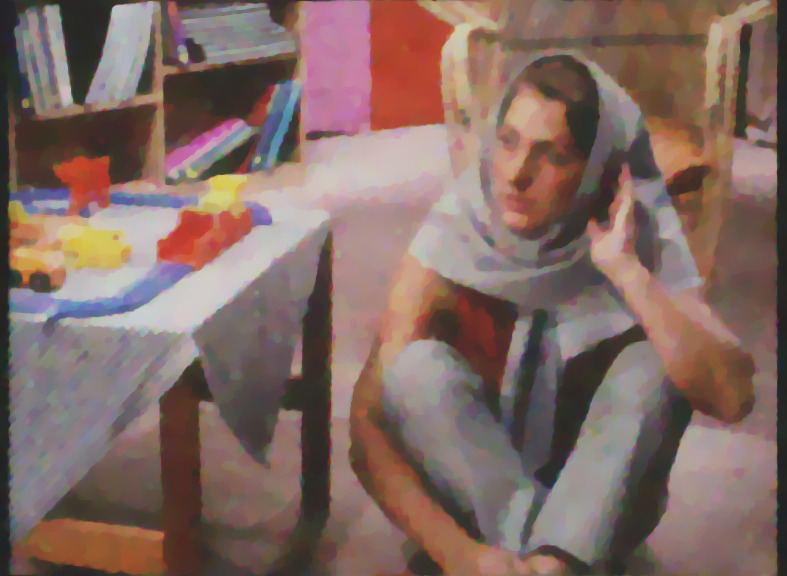}
	\end{minipage}
}
	\subfloat{
	\label{fig:subfig_f}
	\begin{minipage}[t]{0.15\textwidth}
		\centering
		\includegraphics[angle=0,width=1\textwidth]{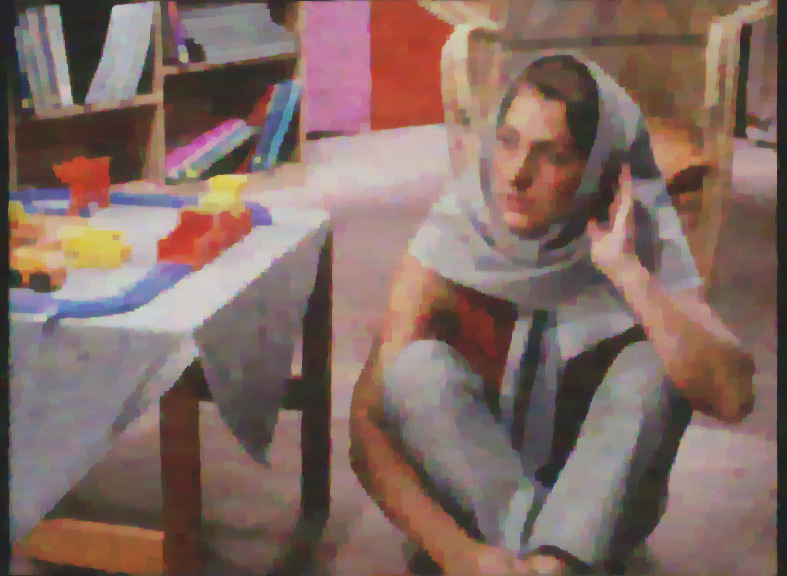}
	\end{minipage}
}
\vfill
	\subfloat{
		\label{fig:subfig_c}
		\begin{minipage}[t]{0.15\textwidth}
			\centering
			\includegraphics[angle=0,width=1\textwidth]{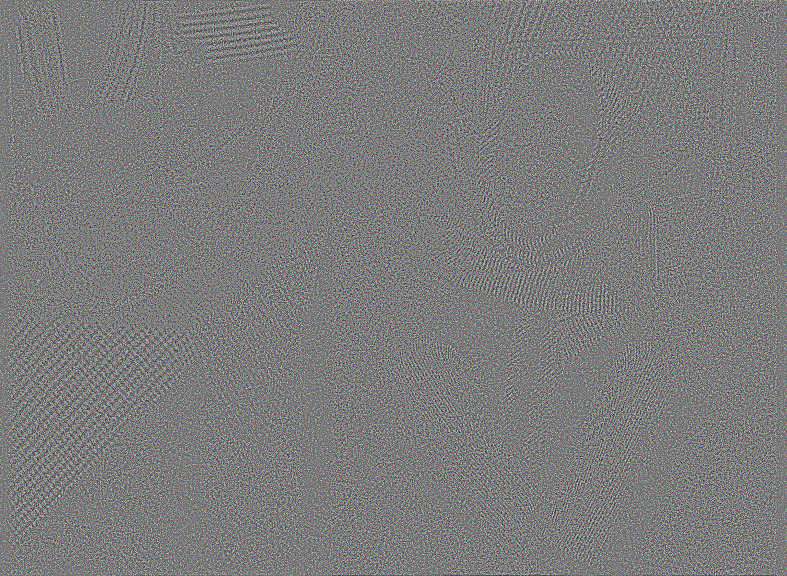}
		\end{minipage}
	}
	\subfloat{
		\label{fig:subfig_g}
		\begin{minipage}[t]{0.15\textwidth}
			\centering
			\includegraphics[angle=0,width=1\textwidth]{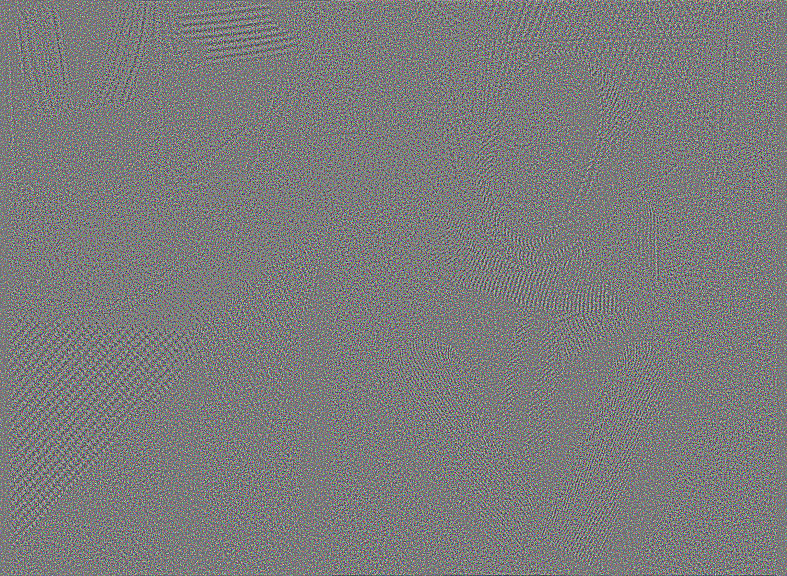}
		\end{minipage}
	}
	\subfloat{
		\label{fig:subfig_g}
		\begin{minipage}[t]{0.15\textwidth}
			\centering
			\includegraphics[angle=0,width=1\textwidth]{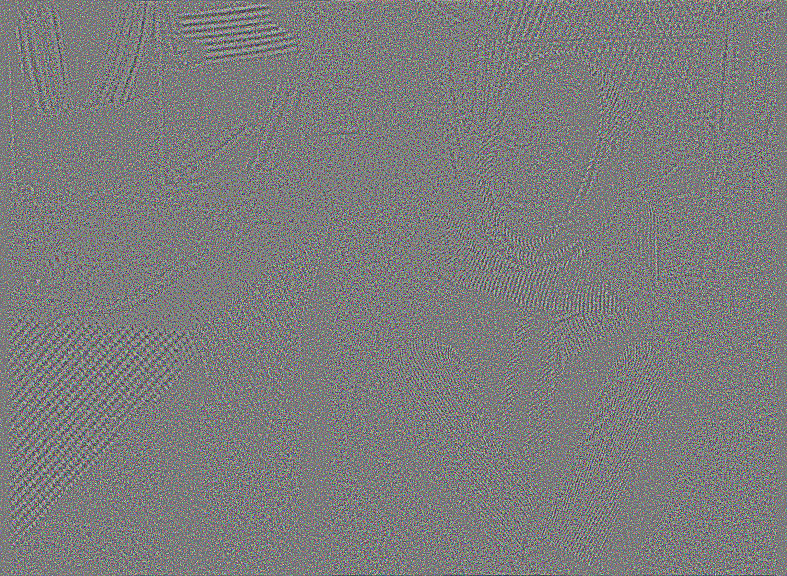}
		\end{minipage}
	}
	\caption{Image decomposition on the noised image (\emph{e}) in Fig.~\ref{test_jpg} with $A=I$. From top to bottom are the noised image, cartoon parts and texture parts. From left to right are the decompositions by the ADME, ADMGB and dADMM.}
	\label{noised_ADME_ADMGB_ADMMD_jpg}
\end{figure}

\begin{figure}[htbp]
	\normalsize
	\subfloat{
		\label{fig:subfig_a}
		\begin{minipage}[t]{0.45\textwidth}
			\centering
			\includegraphics[angle=0,width=1\textwidth]{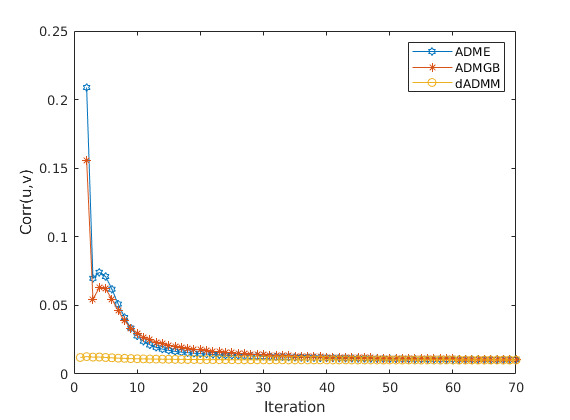}
		\end{minipage}
	}
	\caption{Variations of Corr$(u,v)$ with respect to the iterations for the ADME, ADMGB and dADMM on the image (\emph{e}) in Fig.~\ref{test_jpg} with Gaussian noise.}
	\label{corr_noised_KI_jpg}
\end{figure}

We can see that the dADMM has a better performance than the other two algorithms (i.e., ADME and ADMGB in \cite{Ng_Yuan_Zhang}) in the experiments in case of $A=I$.

\subsection{Case 2: $A=S$}
In this case, $S$ is a blurring matrix, i.e., $A$ in \eqref{eq:general} is a
convolution operator. We test both the Out-of-focus blur and the Gaussian blur in this subsection with the images (\emph{a}) and (\emph{d}) in Fig.~\ref{test_jpg} and compare the performances of the algorithms (the ADME, ADMGB and dADMM). To implement the ADME and ADMGB, we take the parameters as $\tau = 5\times10^{-5}$, $\mu =1\times 10^{-5}$ and $\beta_{1} = \beta_{2} =\beta_{3} = 1\times 10^{-2}$. To implement the dADMM, we set the parameters as $\tau =8 \times 10^{-6}$, $\mu = 4 \times 10^{-4}$ and $\sigma= 2\times10^{2}$.
\begin{figure}[htbp]
	\centering  
	\subfloat{
		\label{fig:subfig_b}
		\begin{minipage}[t]{0.2\textwidth}
			\centering
			\includegraphics[angle=0,width=1\textwidth]{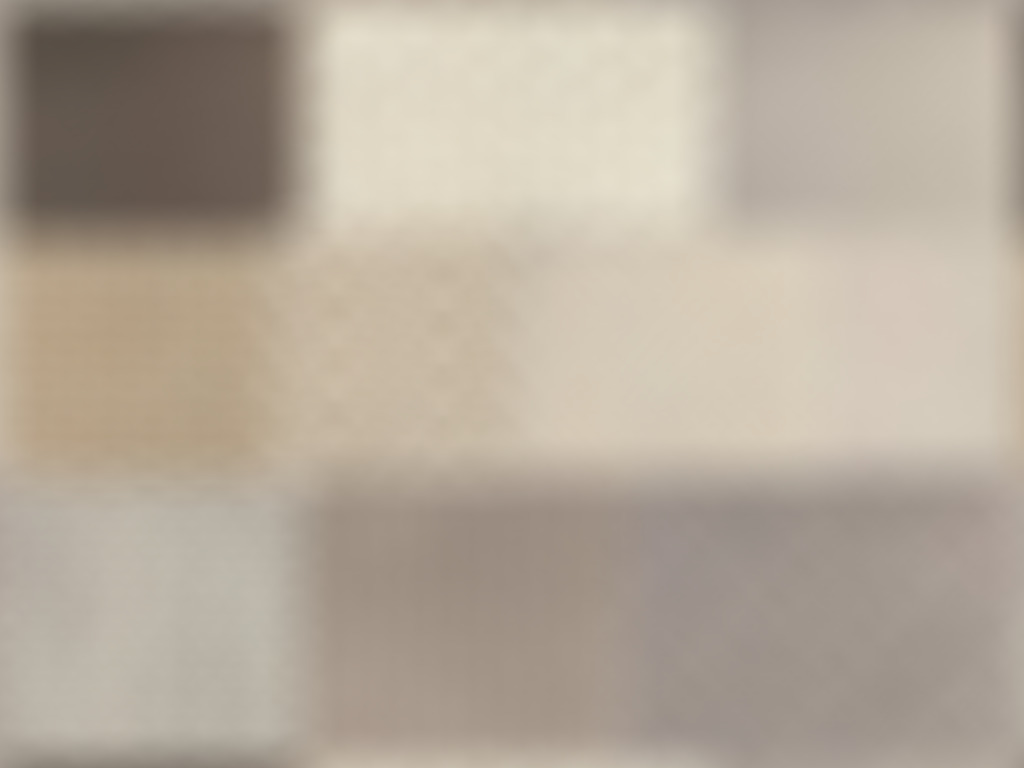}
		\end{minipage}
	}
\vfill 
	\subfloat{
		\label{fig:subfig_c}
		\begin{minipage}[t]{0.2\textwidth}
			\centering
			\includegraphics[angle=0,width=1\textwidth]{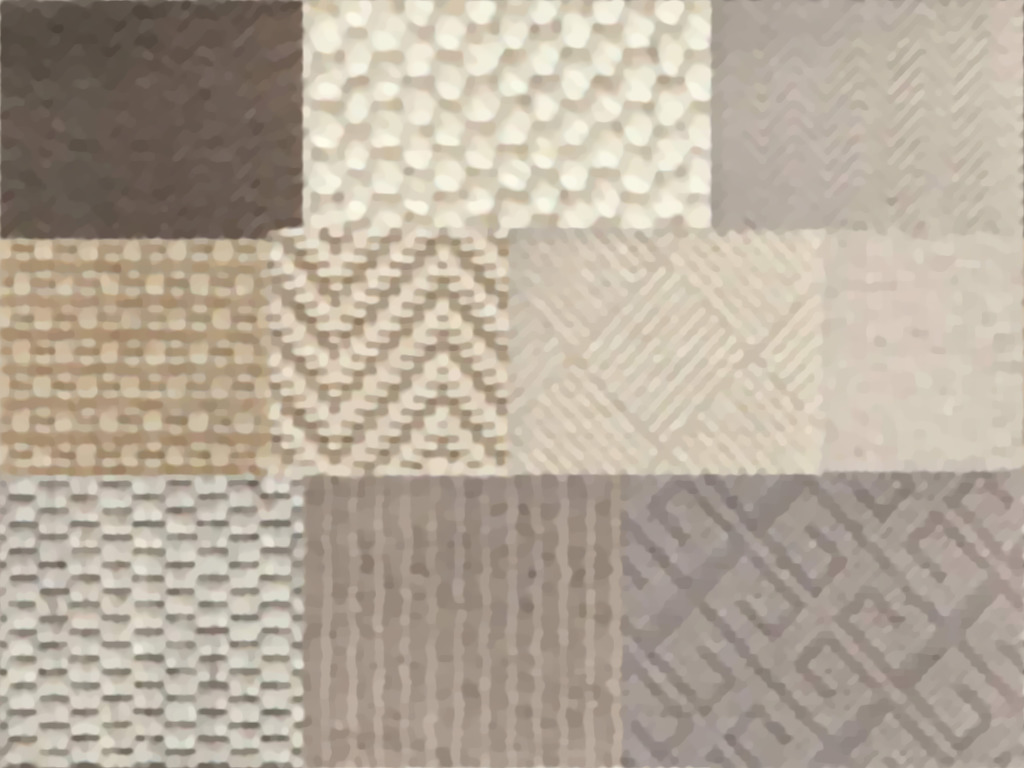}
		\end{minipage}
	}
	\subfloat{
		\label{fig:subfig_d}
		\begin{minipage}[t]{0.2\textwidth}
			\centering
			\includegraphics[angle=0,width=1\textwidth]{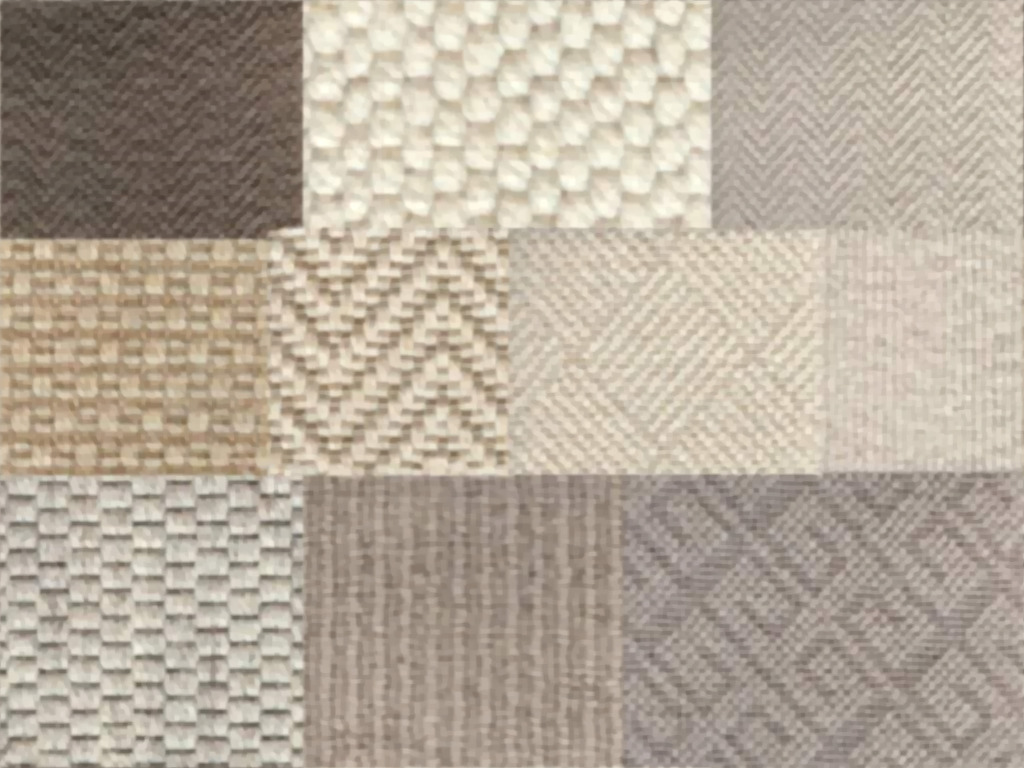}
		\end{minipage}
	}
	\vfill 
	\subfloat{
		\label{fig:subfig_e}
		\begin{minipage}[t]{0.2\textwidth}
			\centering
			\includegraphics[angle=0,width=1\textwidth]{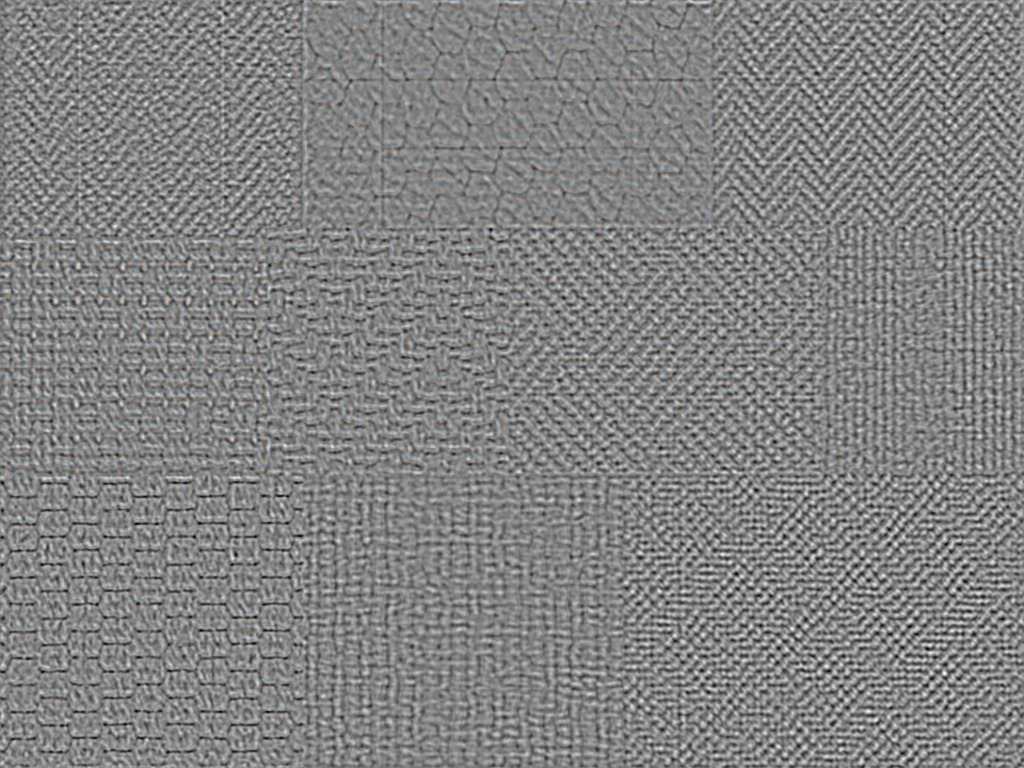}
		\end{minipage}
	}
	\subfloat{
		\label{fig:subfig_f}
		\begin{minipage}[t]{0.2\textwidth}
			\centering
			\includegraphics[angle=0,width=1\textwidth]{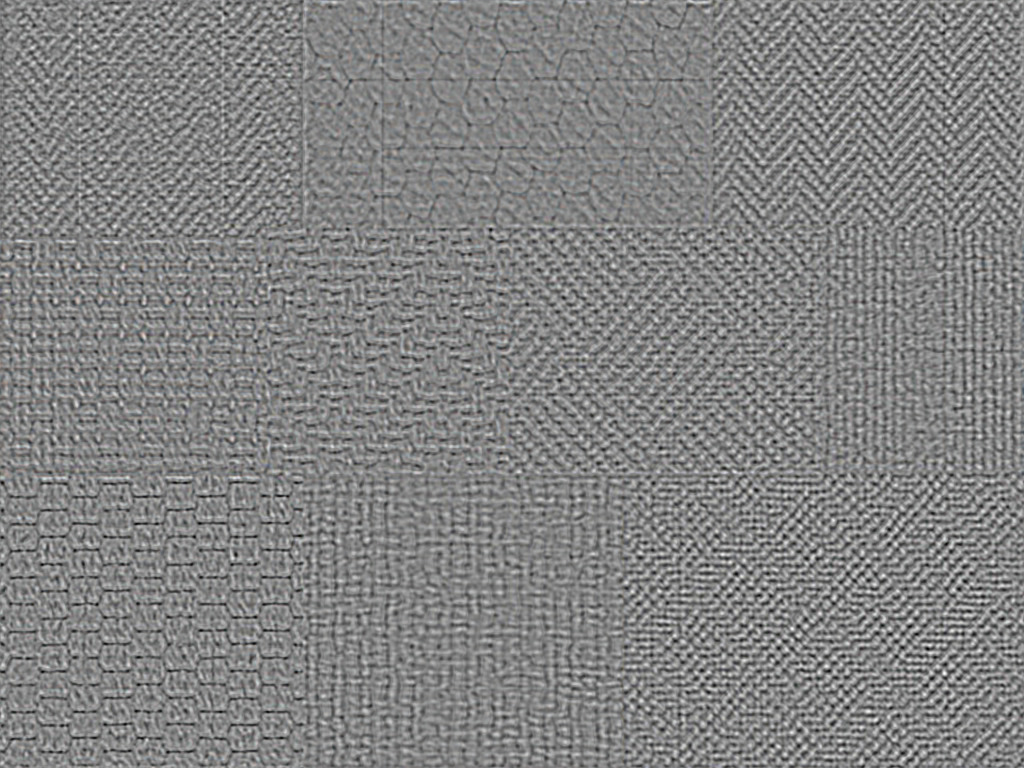}
		\end{minipage}
	}
\vfill 
	\subfloat{
		\label{fig:subfig_g}
		\begin{minipage}[t]{0.2\textwidth}
			\centering
			\includegraphics[angle=0,width=1\textwidth]{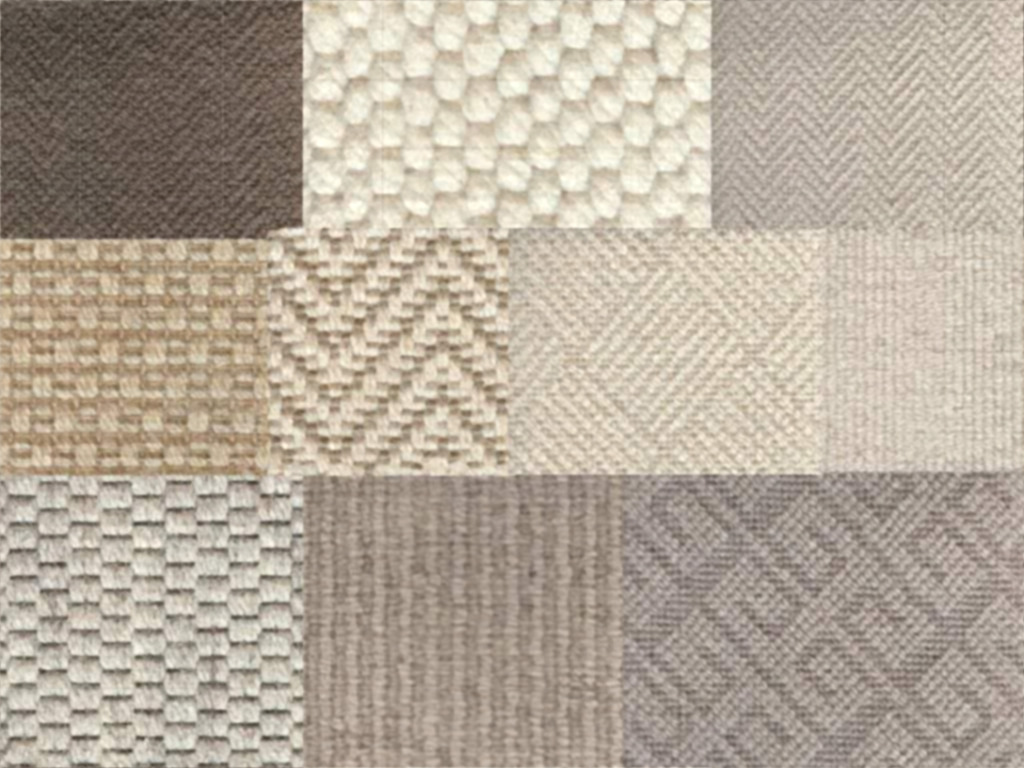}
		\end{minipage}
	}
	\subfloat{
		\label{fig:subfig_h}
		\begin{minipage}[t]{0.2\textwidth}
			\centering
			\includegraphics[angle=0,width=1\textwidth]{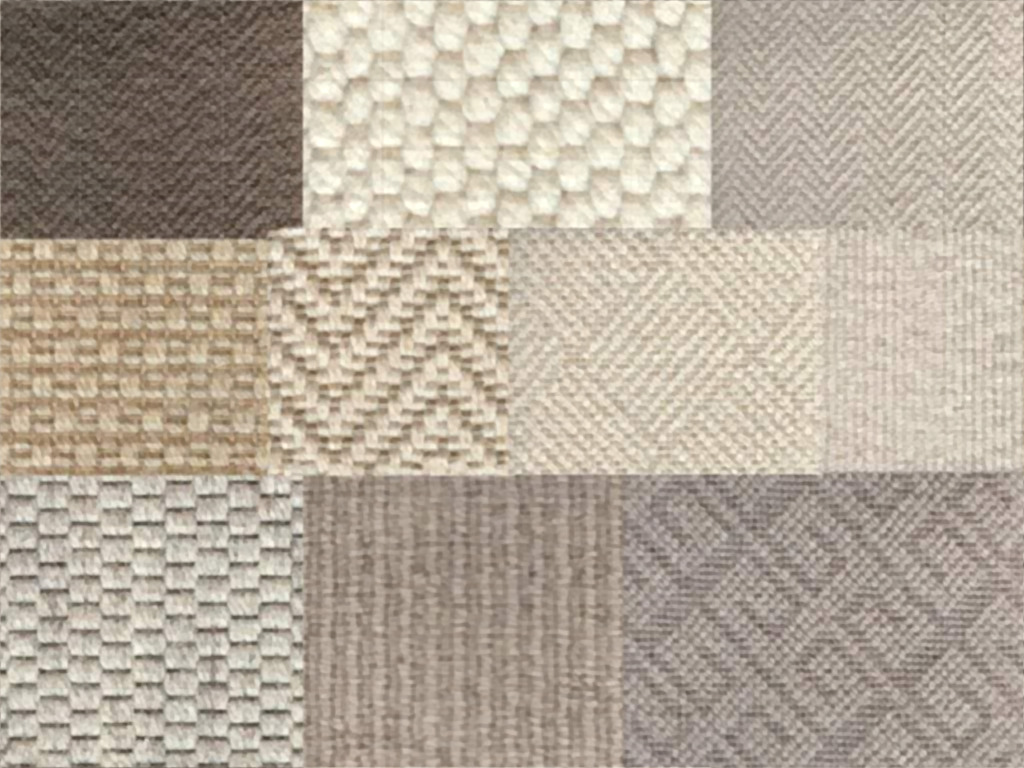}
		\end{minipage}
	}
	
	\caption{Image decomposition and restoration on \emph{blurred} image with ``Out-of-focus(40)'' about image (\emph{d}) in Fig.~\ref{test_jpg}. From top to bottom are the blurred image, cartoons, textures and restored images.  From left to right are the results obtained by the ADMGB and dADMM, respectively.}
	\label{de_re_deblurring_jpg}
\end{figure}
\renewcommand{\arraystretch}{1.1}
\begin{table*}[htbp]	
	\centering
	\fontsize{7}{8}\selectfont
	\caption{Image decomposition on the images with blurry: ``Iter'' -- the number of iterations; ``Tol'' : -- tolerance for the stopping criterion; ``Time'' -- computing time (in seconds); ``$\mbox{PSNR}^{0}$'' (``PSNR'') -- the PSNR value between the blurred (restored) image and the original image, respectively. And a=``ADME''; b=``ADMGB''; c= ``dADMM''.}
	\label{deblurring_ADME_ADMGB_ADMMD_table}
	\begin{tabular}{|c||c|c||c|c|c|c|}
		\hline
		Image&Blur&$\mbox{PSNR}^{0}$&Iter (a$\,|\,$b$\,|\,$c)&Tol (a$\,|\,$b$\,|\,$c)&Time (a$\,|\,$b$\,|\,$c)&PSNR (a$\,|\,$b$\,|\,$c)\cr\hline
		\multirow{4}{*}{\emph{Lena}}
		&Gaussian(20,20)&22.25&70$\,|\,$70$\,|\,$23&3.0e-1$\,|\,$3.0e-1$\,|\,$9.8e-4&5.19$\,|\,$5.69$\,|\,$1.39&30.41$\,|\,$30.42$\,|\,$31.38\\
		&Gaussian(30,30)&20.61&70$\,|\,$70$\,|\,$21&2.1e-1$\,|\,$2.1e-1$\,|\,$8.6e-4&5.20$\,|\,$5.59$\,|\,$1.25&28.23$\,|\,$28.25$\,|\,$28.87\\
		&Out-of-focus(20)&19.95&70$\,|\,$70$\,|\,$19&3.5e-3$\,|\,$1.0e-2$\,|\,$9.5e-4&5.48$\,|\,$5.84$\,|\,$1.12&29.80$\,|\,$29.82$\,|\,$30.73\\
		&Out-of-focus(30)&18.48&70$\,|\,$70$\,|\,$16&3.4e-3$\,|\,$2.7e-2$\,|\,$9.2e-4&5.20$\,|\,$5.73$\,|\,$0.97&27.76$\,|\,$27.78$\,|\,$28.51\\ \hline
		\multirow{4}{*}{\emph{Weave}}
		&Gaussian(40,40)&19.60&70$\,|\,$70$\,|\,$49&2.2e-1$\,|\,$2.2e-1$\,|\,$9.9e-4&77.04$\,|\,$86.54$\,|\,$66.73&26.40$\,|\,$26.42$\,|\,$27.21\\
		&Gaussian(50,50)&19.44&70$\,|\,$70$\,|\,$47&1.9e-1$\,|\,$1.9e-1$\,|\,$9.8e-4&76.85$\,|\,$86.07$\,|\,$64.16&24.95$\,|\,$24.96$\,|\,$26.68\\
		&Out-of-focus(40)&19.30&70$\,|\,$70$\,|\,$39&1.7e-2$\,|\,$4.3e-2$\,|\,$9.9e-4&75.86$\,|\,$86.10$\,|\,$49.23&26.85$\,|\,$26.87$\,|\,$27.05\\
		&Out-of-focus(55)&19.10&70$\,|\,$70$\,|\,$39&1.3e-2$\,|\,$3.8e-2$\,|\,$9.9e-4&76.89$\,|\,$61.87$\,|\,$49.15&24.78$\,|\,$24.80$\,|\,$26.51\\ \hline
	\end{tabular}
\end{table*}

The numerical results of the ADME, ADMGB and dADMM on the images (\emph{a}) and (\emph{d}) with different blur kernels are showed in Table~\ref{deblurring_ADME_ADMGB_ADMMD_table}. In the table, ``Gaussian(40,40)'' means the blur kernel is generated by {\sc Matlab} function \texttt{fspecial(`gaussian',40,40)}, and ``Out-of-focus(40)'' means the Out-of-focus blur kernel with a radius of $40$ given by {\sc Matlab} function \texttt{fspecial('disk',40)}.

In Table~\ref{deblurring_ADME_ADMGB_ADMMD_table} we report the detailed numerical results for the ADME, ADMGB and dADMM in the image decomposition and reconstruction problems. One can observe from Table~\ref{deblurring_ADME_ADMGB_ADMMD_table} that the dADMM takes less time and less iterations to get a higher PSNR value of the restored image than the ADME and ADMGB for most of the tested examples. And we display the decomposed results of the cartoon parts, the texture parts and the restored images of \emph{Weave} image ((\emph{d}) in Fig.~\ref{test_jpg}) in the case of ``Out-of-focus(40)'' in Table \ref{deblurring_ADME_ADMGB_ADMMD_table} by two algorithms (ADMGB and dADMM) in Fig.~\ref{de_re_deblurring_jpg}. The numerical performances indicate that the dADMM is a robust, high-performance algorithm for the optimization problem \eqref{eq:target}.

\subsection{Case 3: $A=K$}
This subsection is devoted to the more difficult part of image decomposition and restoration on an image with missing pixels. For a binary matrix $K\in \textbf{R}^{m\times n}$ and a clean image $M\in \textbf{R}^{m\times n}$. The degraded image $D = K\circ M$, where the operator $\circ$ denotes the Hadamard product, i.e., $D_{i,j} = K_{i,j}\times M_{i,j}$ ($i=1,\dots, m;\,j=1,\dots,n$). The images (\emph{f}) and (\emph{h}) in Fig.~\ref{test_jpg} are used in this subsection. And $K$ is set to be a $512\times 512$ matrix and a $1024\times 1024$ matrix for the image (\emph{f}) and (\emph{h}), respectively.

\begin{figure}[htbp]
	\centering  
	\normalsize
	\subfloat{
		\label{fig:subfig_a}
		\begin{minipage}[t]{0.45\textwidth}
			\centering
			\includegraphics[angle=0,width=1\textwidth]{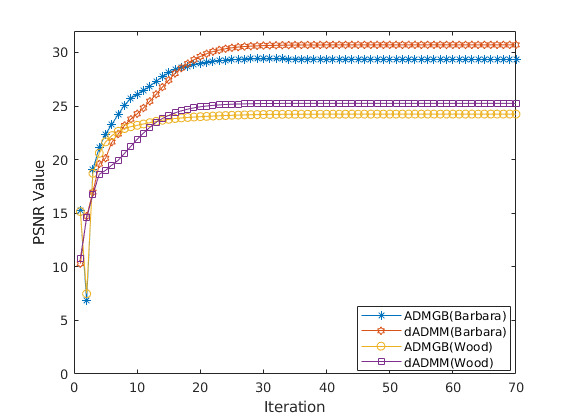}
		\end{minipage}
	}
	\caption{Variations of the PSNR values with respect to the iterations for the ADMGB and dADMM on the images (\emph{f}) and (\emph{h}) in Fig.~\ref{test_jpg} with missing pixels.}
	\label{psnr_KS_jpg}
\end{figure}

\renewcommand{\arraystretch}{1.3}
\begin{table}[htbp]	
	\centering
	\fontsize{7}{8}\selectfont
	\caption{Image decomposition and restoration on the images with missing pixels:  ``Iter'' -- the number of iterations; ``Tol'' -- tolerance for the stopping criterion; ``Time'' -- computing time (in seconds); ``$\mbox{PSNR}^{0}$'' (``PSNR'') -- the PSNR value between the painted (restored) image and the original image, respectively. And a=``ADMGB''; b= ``dADMM''.}
	\label{inpainting_ADMGB_ADMMD_table}
	\begin{tabular}{|c|c||c|c|c|c|}
		\hline
		Image&$\mbox{PSNR}^{0}$&Iter (a$|$b)&Tol (a$|$b)&Time (a$|$b)&PSNR (a$|$b)\cr\hline
		\emph{Barbara}&14.13&70$|$39&7.5e-3$|$9.3e-4&4.95$|$3.23&29.34$|$30.71\\
		\emph{Wood}&14.81&70$|$39&2.1e-2$|$9.0e-4&97.63$|$79.16&24.24$|$25.25\\
		\hline
	\end{tabular}
\end{table}

\begin{figure*}[htbp]
	\centering
	\subfloat{
		\label{fig:subfig_e}
		\begin{minipage}[t]{0.2\textwidth}
			\centering
			\includegraphics[angle=0,width=1\textwidth]{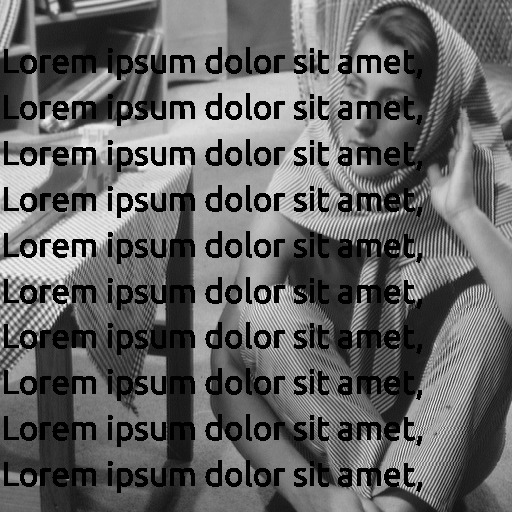}
		\end{minipage}
	}
	\subfloat{
	\label{fig:subfig_e}
	\begin{minipage}[t]{0.2\textwidth}
		\centering
		\includegraphics[angle=0,width=1\textwidth]{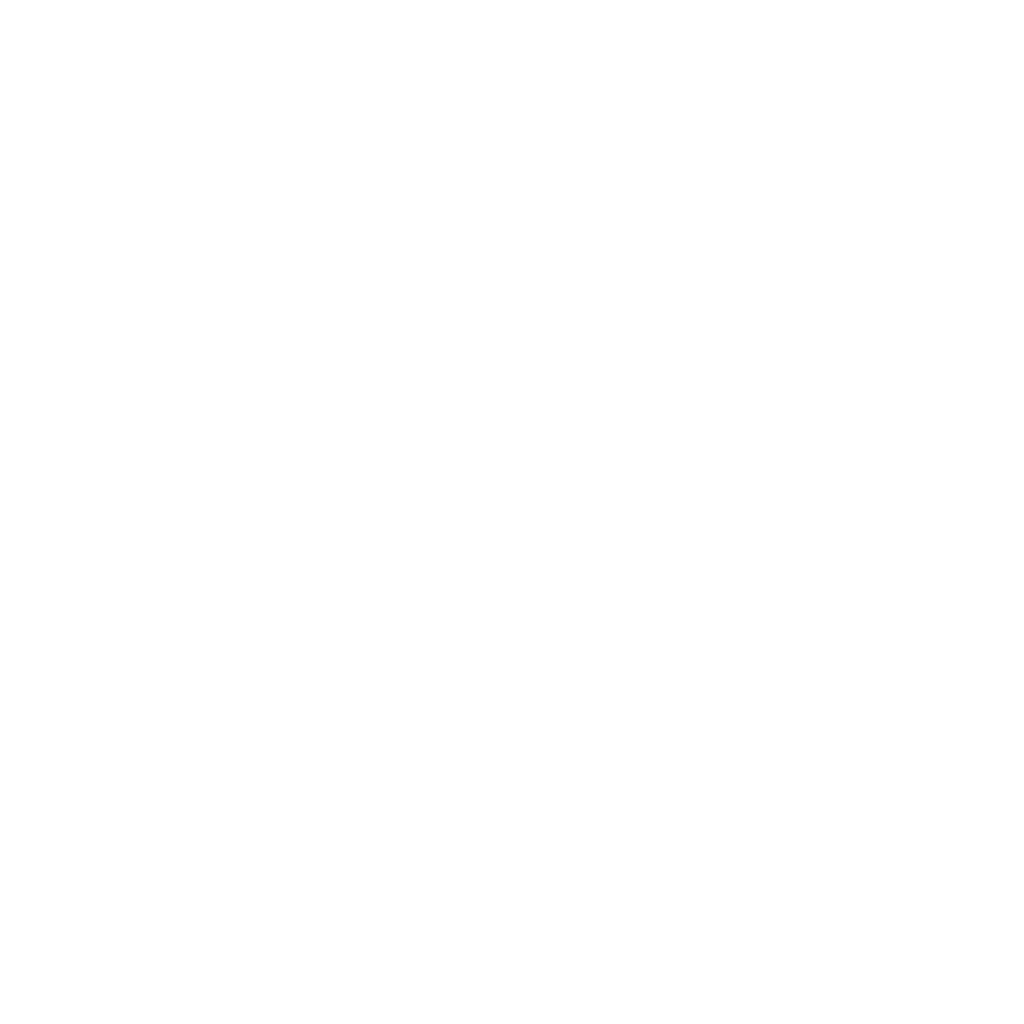}
	\end{minipage}
}	
\subfloat{
	\label{fig:subfig_e}
	\begin{minipage}[t]{0.2\textwidth}
		\centering
		\includegraphics[angle=0,width=1\textwidth]{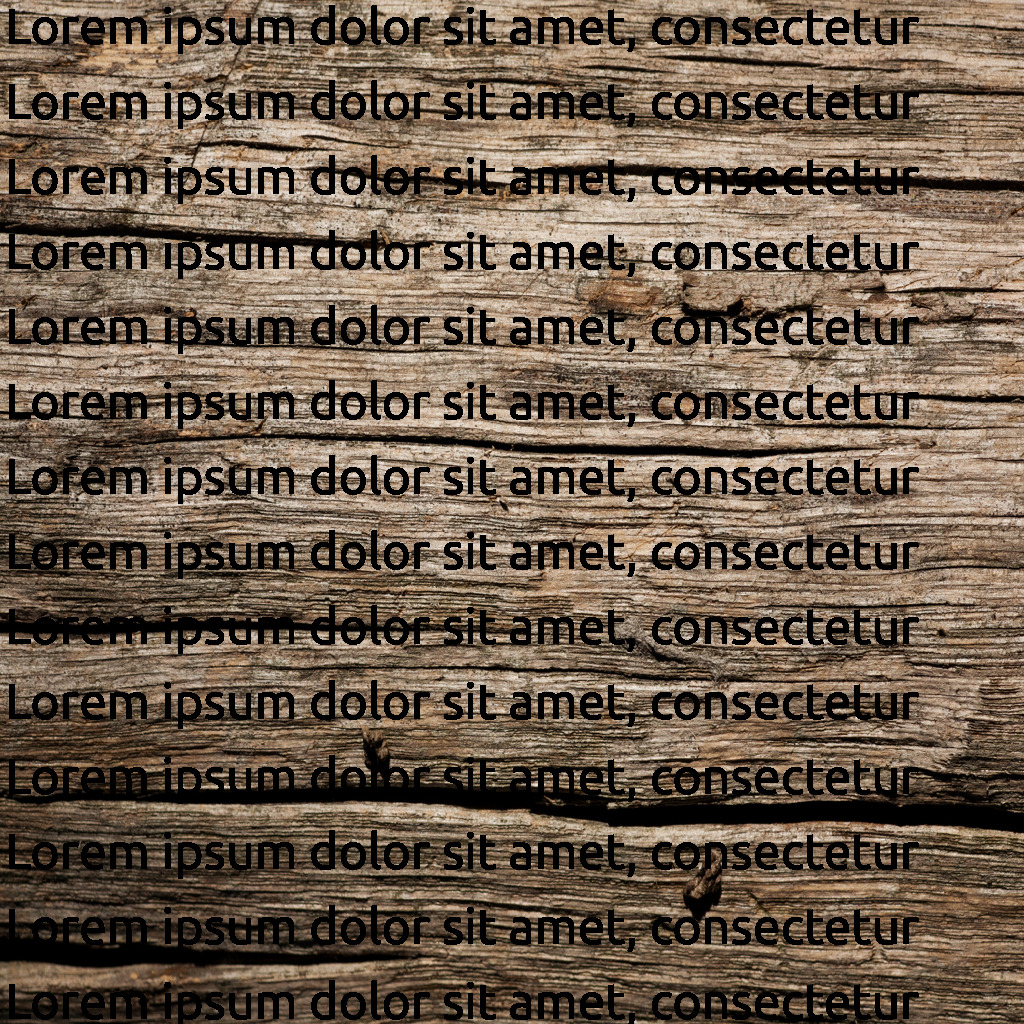}
	\end{minipage}
}
\vfill
	\subfloat{
		\label{fig:subfig_f}
		\begin{minipage}[t]{0.2\textwidth}
			\centering
			\includegraphics[angle=0,width=1\textwidth]{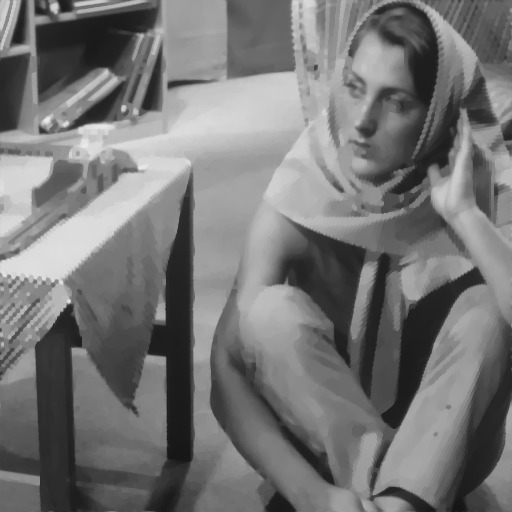}
		\end{minipage}
	}
	\subfloat{
	\label{fig:subfig_b}
	\begin{minipage}[t]{0.2\textwidth}
		\centering
		\includegraphics[angle=0,width=1\textwidth]{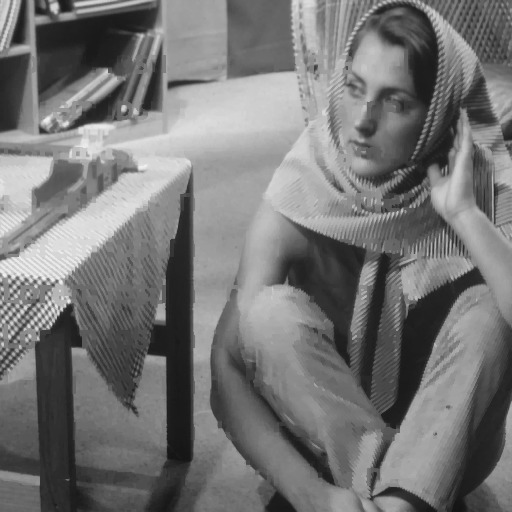}
	\end{minipage}
}
	\subfloat{
	\label{fig:subfig_f}
	\begin{minipage}[t]{0.2\textwidth}
		\centering
		\includegraphics[angle=0,width=1\textwidth]{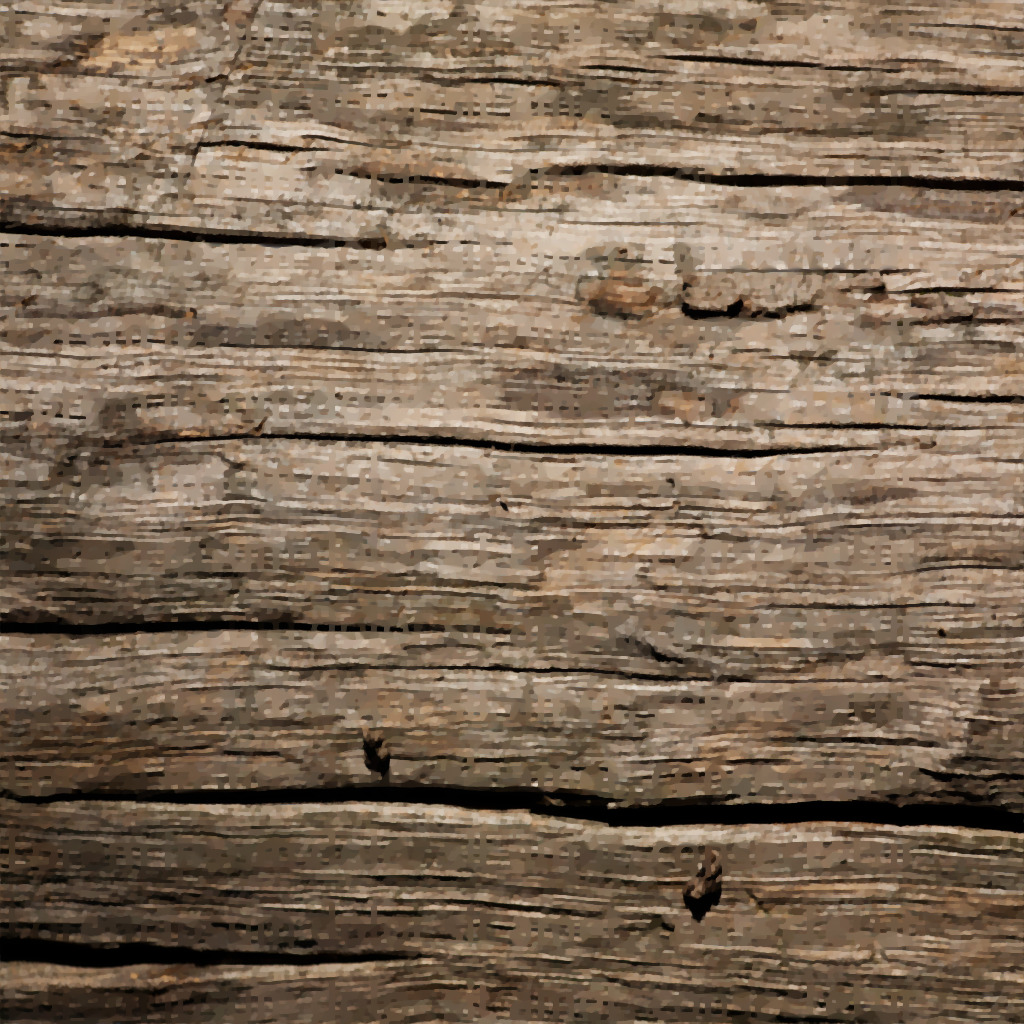}
	\end{minipage}
}
\subfloat{
	\label{fig:subfig_f}
	\begin{minipage}[t]{0.2\textwidth}
		\centering
		\includegraphics[angle=0,width=1\textwidth]{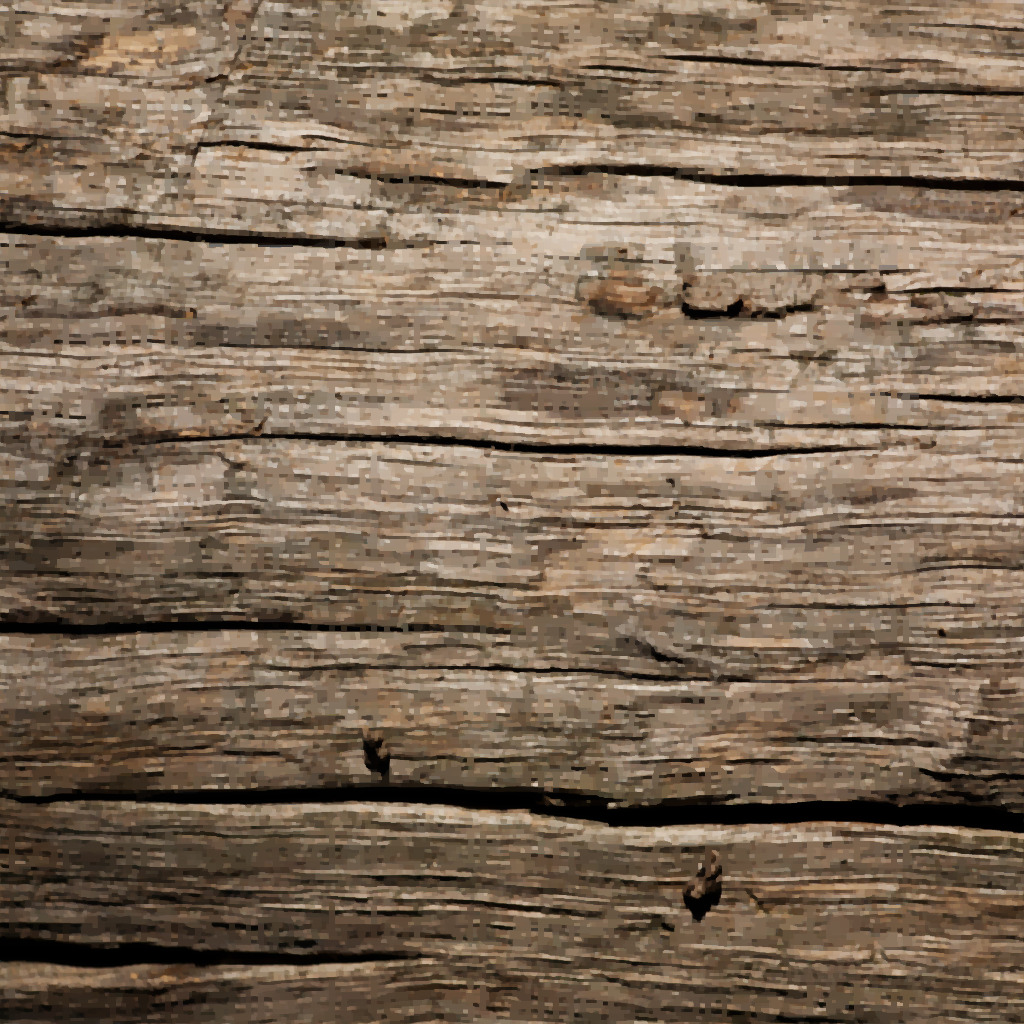}
	\end{minipage}
}
\vfill
	\subfloat{
		\label{fig:subfig_g}
		\begin{minipage}[t]{0.2\textwidth}
			\centering
			\includegraphics[angle=0,width=1\textwidth]{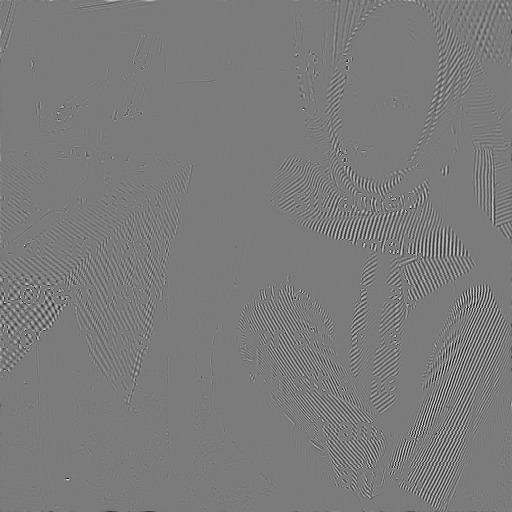}
		\end{minipage}
	}
	\subfloat{
	\label{fig:subfig_c}
	\begin{minipage}[t]{0.2\textwidth}
		\centering
		\includegraphics[angle=0,width=1\textwidth]{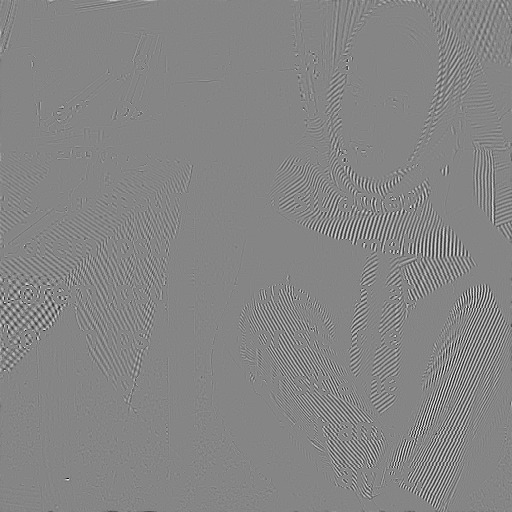}
	\end{minipage}
}
	\subfloat{
	\label{fig:subfig_g}
	\begin{minipage}[t]{0.2\textwidth}
		\centering
		\includegraphics[angle=0,width=1\textwidth]{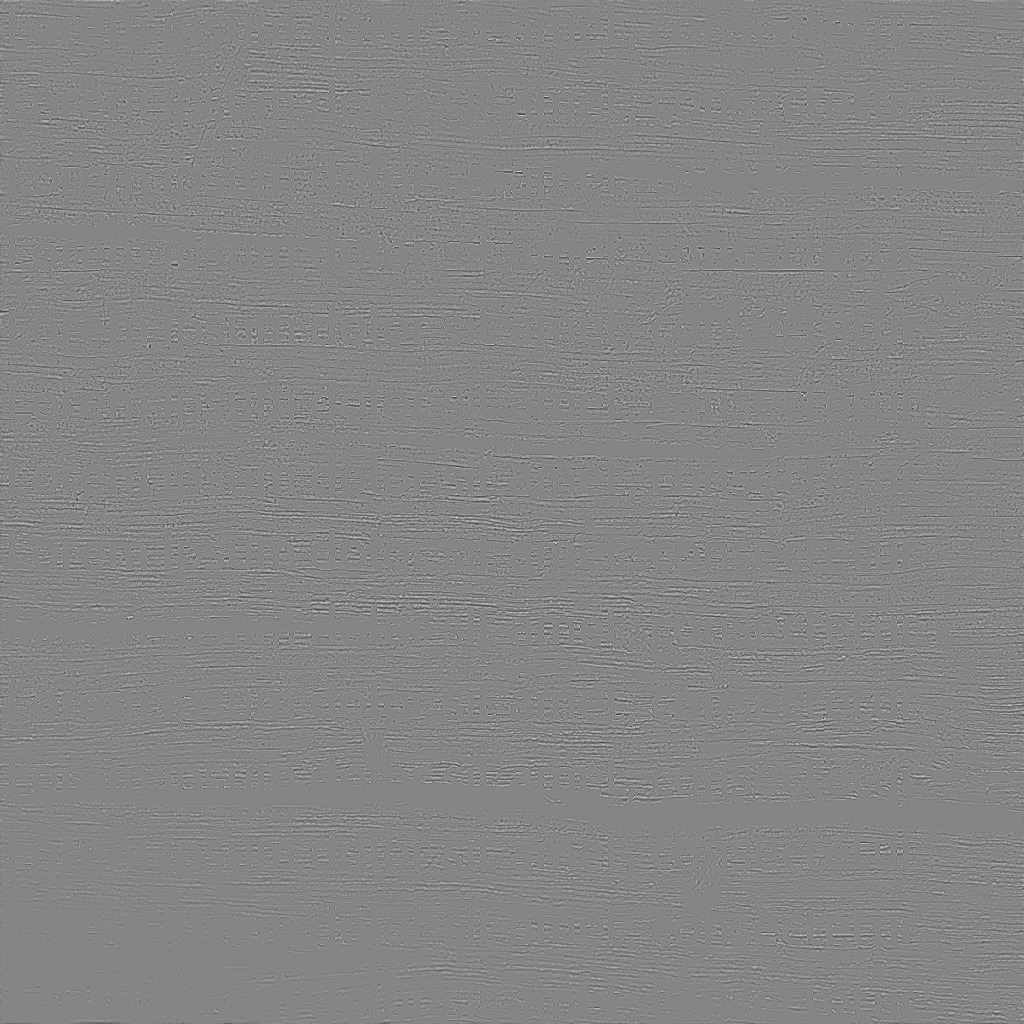}
	\end{minipage}
}
\subfloat{
	\label{fig:subfig_g}
	\begin{minipage}[t]{0.2\textwidth}
		\centering
		\includegraphics[angle=0,width=1\textwidth]{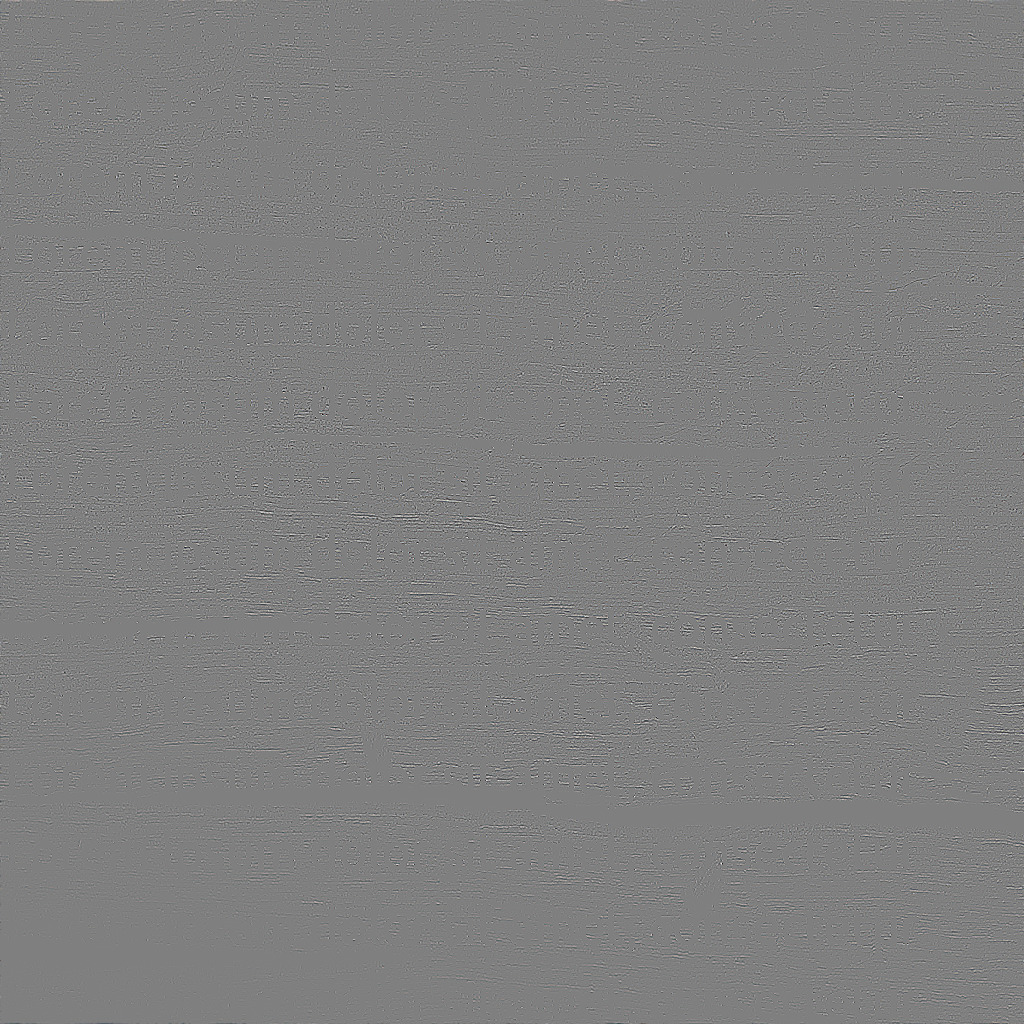}
	\end{minipage}
}
\vfill
	\subfloat{
		\label{fig:subfig_h}
		\begin{minipage}[t]{0.2\textwidth}
			\centering
			\includegraphics[angle=0,width=1\textwidth]{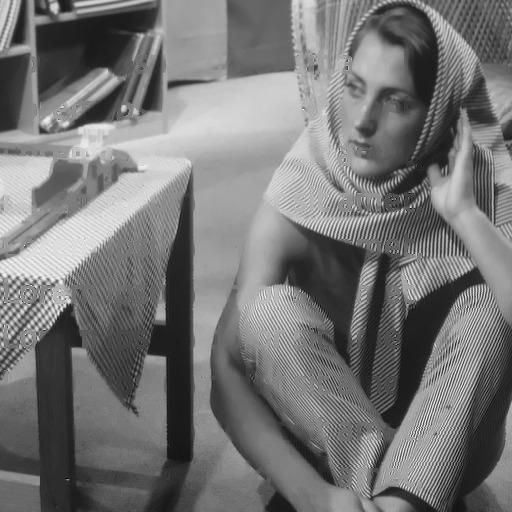}
		\end{minipage}
	}
	\subfloat{
		\label{fig:subfig_d}
		\begin{minipage}[t]{0.2\textwidth}
			\centering
			\includegraphics[angle=0,width=1\textwidth]{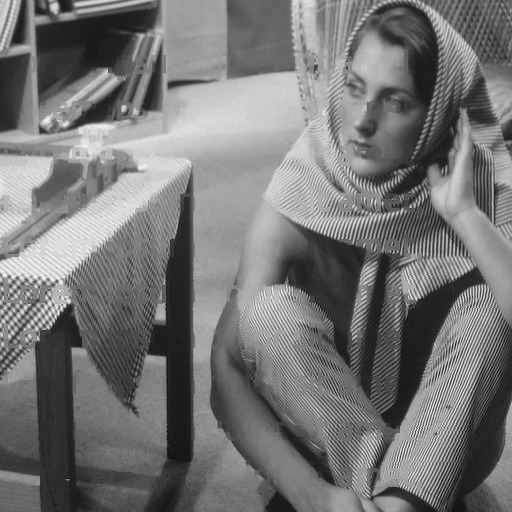}
		\end{minipage}
	}
	\subfloat{
		\label{fig:subfig_h}
		\begin{minipage}[t]{0.2\textwidth}
			\centering
			\includegraphics[angle=0,width=1\textwidth]{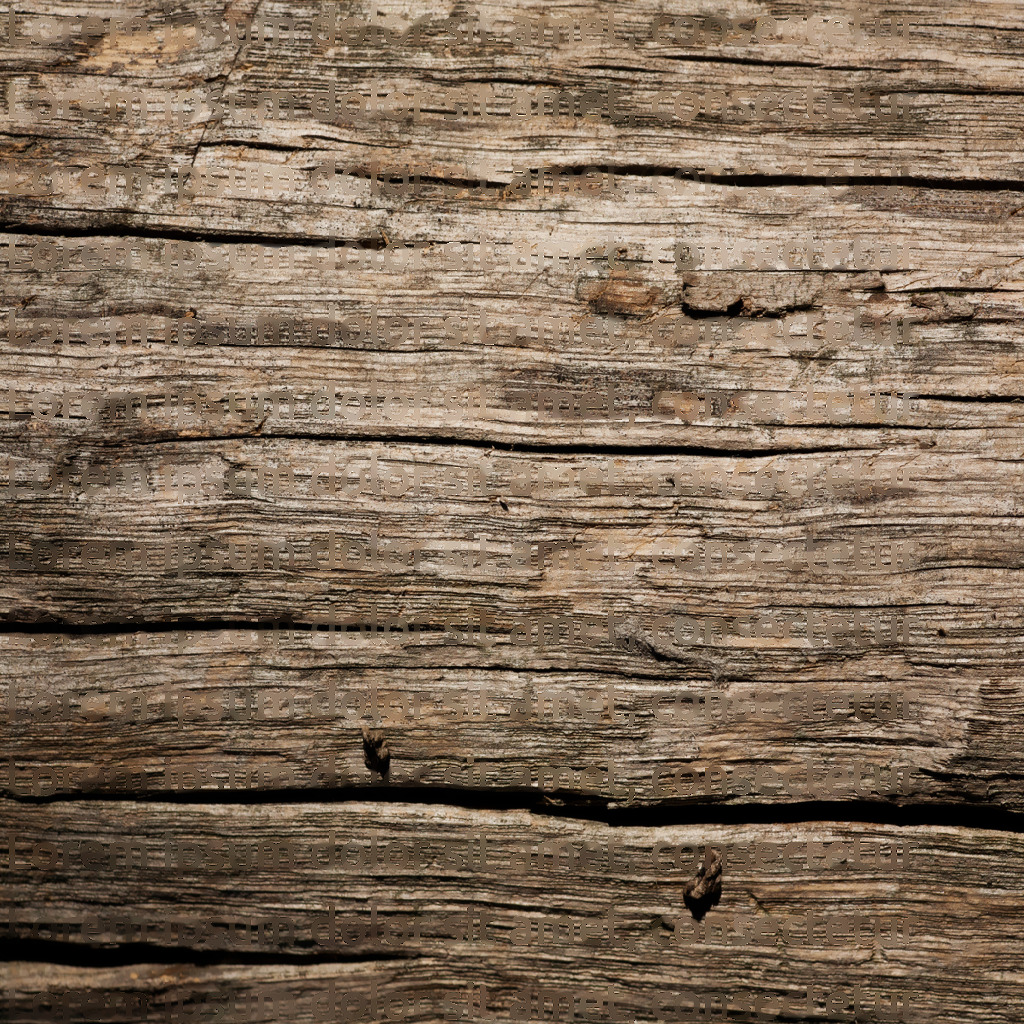}
		\end{minipage}
	}
\subfloat{
	\label{fig:subfig_h}
	\begin{minipage}[t]{0.2\textwidth}
		\centering
		\includegraphics[angle=0,width=1\textwidth]{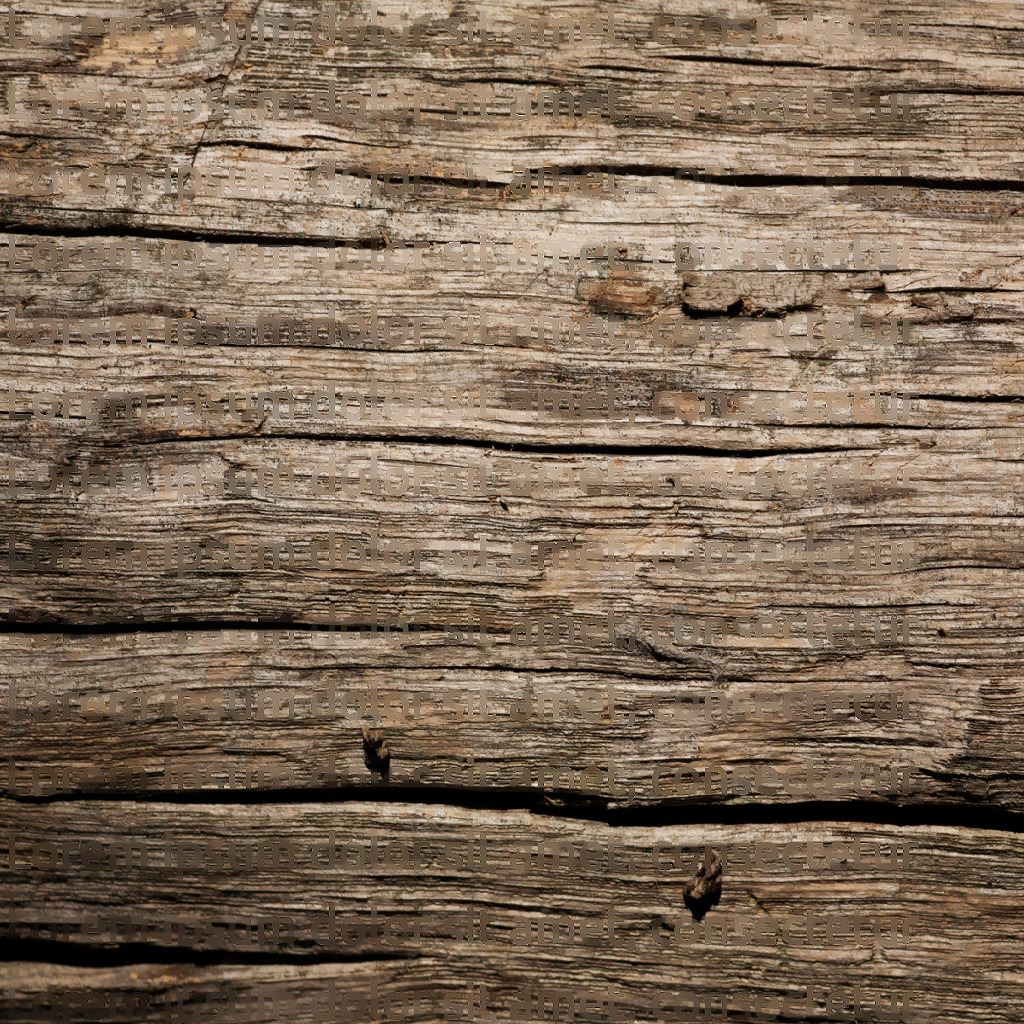}
	\end{minipage}
}
	\caption{Image decomposition and restoration on images with missing pixels. From top row to bottom row are degraded images, cartoon parts, texture parts and restored images. $512\times 512$ \emph{Barbara} image with missing pixels, decomposition and restoration  image by the ADMGB  and dADMM on first two column respectively. $1024\times 1024 \times 3$ \emph{Wood} image with missing pixels, decomposition and restoration image by the ADMGB and dADMM on last two column respectively.}
	\label{de_re_inpainted_jpg}
\end{figure*}

We compare the performances of the ADMGB and dADMM in this case.
The parameters of the ADMGB are set as $\tau = 1\times 10^{-2}$, $\mu = 5\times 10^{-3}$ and $\beta_{1} = \beta_{2} =\beta_{3} = 5\times10^{-2}$. To implement the dADMM, we take $\tau=4\times10^{-3}$, $\mu = 1\times10^{-3}$, and $\sigma = 3\times10^{3}$. Finally, the PSNR values with respect to the iterations for the images ((\emph{f}) and (\emph{h}) in Fig.~\ref{test_jpg}) with missing pixels are showed in Fig.~\ref{psnr_KS_jpg}. We display the decomposed images and the restored images of two algorithms (ADMGB and dADMM) for the images ((\emph{f}) and (\emph{h}) in Fig.~\ref{test_jpg}) in Fig.~\ref{de_re_inpainted_jpg}.

From Fig.~\ref{psnr_KS_jpg} and Table~\ref{inpainting_ADMGB_ADMMD_table} we can see that the PSNR value of the reconstructed image generated by the dADMM is higher than that generated by the ADMGB. Moreover, we can get a better PSNR value with less time and less iterations by the dADMM.

\subsection{Case 4: $A=KS$}
In this subsection, we solve the problem \eqref{eq:general} with $A=KS$, where $S$ is a blurring matrix, and $K$ is a binary matrix. That is, we consider decomposing images with both blurry and missing pixels. The images (\emph{a}) and (\emph{h}) in Fig.~\ref{test_jpg} are used in this part. The blurring matrices $S$ (the Out-of-focus blur and the Gaussian blur) in  Table~\ref{deblurring_ADME_ADMGB_ADMMD_table} are used as blurring matrix, and $K$ is set to be a $512\times 512$ matrix and a $1024\times 1024$ matrix for image (\emph{f}) and (\emph{h}), respectively.

\renewcommand{\arraystretch}{1.2}
\begin{table*}[htbp]	
	\centering
	\fontsize{7}{8}\selectfont
	\caption{Image decomposition and restoration of the images with blurry and missing pixels: ``Iter'' -- the number of iterations; ``Tol'' -- tolerance for the stopping criterion; ``Time'' -- computing time (in seconds); ``$\mbox{PSNR}^{0}$'' (``PSNR'') -- the PSNR value between the blurred (restored) image and the original image, respectively. And  a=``ADMGB''; b= ``dADMM''.}
	\label{deblurring_inpainting_ADMGB_ADMMD_table}
	\begin{tabular}{|c|c|c||c|c|c|c|}
		\hline
		Image&Blur + Missing pixels&$\mbox{PSNR}^{0}$&Iter (a$\,|\,$b)&Tol (a$\,|\,$b)&Time (a$\,|\,$b)&PSNR (a$\,|\,$b)\cr\hline
		\multirow{4}{*}{\emph{Lena}}
		&Gaussian(15,15) + Missing pixels  &11.54&70$\,|\,$16&4.9e-2$\,|\,$7.5e-4&16.20$\,|\,$2.77&24.88$\,|\,$26.44\\
		&Gaussian(20,20) + Missing pixels  &11.47&70$\,|\,$15&5.6e-2$\,|\,$8.2e-4&16.22$\,|\,$2.62 &25.48$\,|\,$26.20\\
		&Out-of-focus(10) + Missing pixels&12.98&70$\,|\,$15&5.3e-2$\,|\,$7.8e-4&15.43$\,|\,$2.59 &25.76$\,|\,$26.76\\
		&Out-of-focus(15) + Missing pixels &11.38&70$\,|\,$15&5.8e-2$\,|\,$7.2e-4&14.85$\,|\,$2.56 &24.05$\,|\,$25.67\\ \hline
		\multirow{4}{*}{\emph{Wood}}
		&Gaussian(20,20) + Missing pixels  &12.70&70$\,|\,$15&8.6e-2$\,|\,$9.1e-4&423.97$\,|\,$101.35&18.54$\,|\,$20.93\\
		&Gaussian(30,30) + Missing pixels  &12.48&70$\,|\,$15&8.5e-2$\,|\,$9.1e-4&441.78$\,|\,$100.35&17.75$\,|\,$19.66\\
		&Out-of-focus(15) + Missing pixels&12.58&70$\,|\,$15&6.8e-2$\,|\,$9.1e-4&395.44$\,|\,$100.55&17.48$\,|\,$19.23\\
		&Out-of-focus(20) + Missing pixels&12.43&70$\,|\,$15&7.2e-2$\,|\,$9.0e-4&395.85$\,|\,$99.96&17.04$\,|\,$18.35\\
		\hline
	\end{tabular}
\end{table*}
\begin{figure}[htbp]
	\centering  
	\subfloat{
		\label{fig:subfig_e}
		\begin{minipage}[t]{0.2\textwidth}
			\centering
			\includegraphics[angle=0,width=1\textwidth]{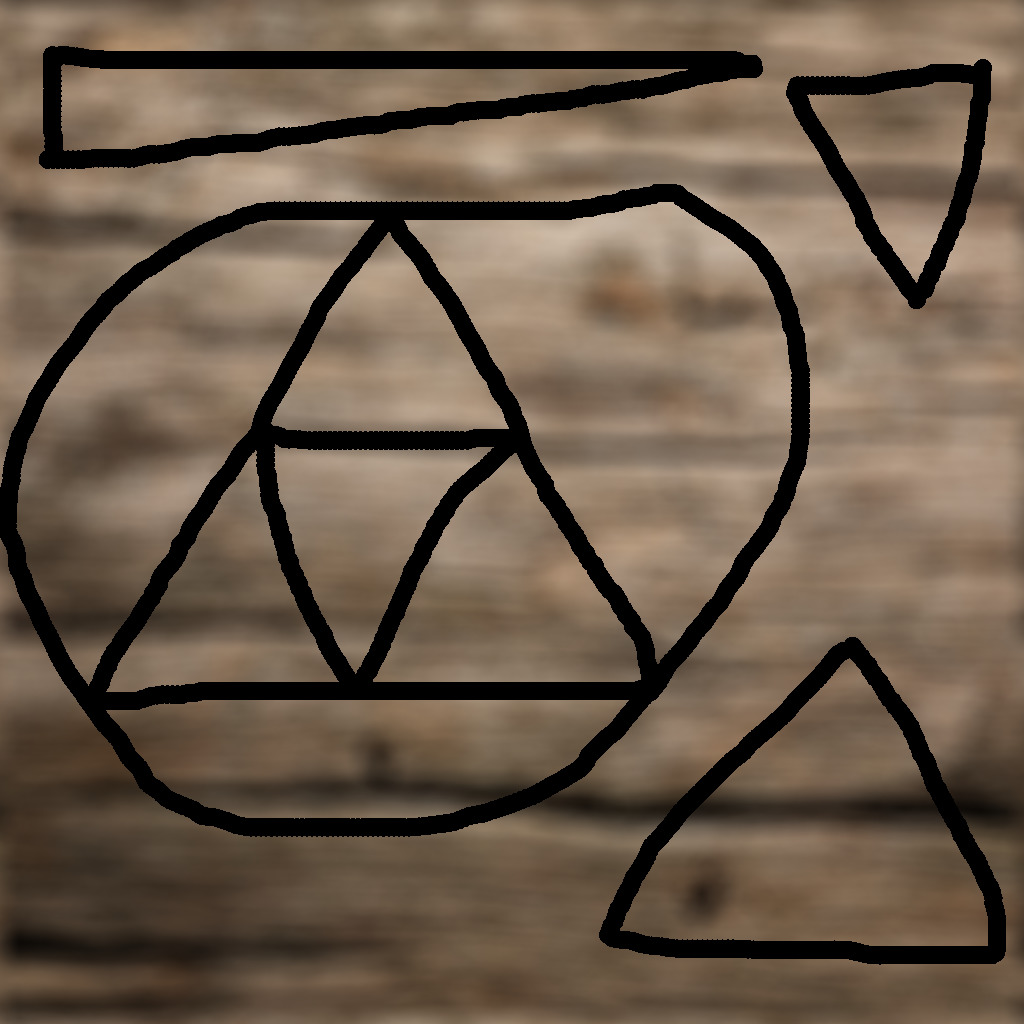}
		\end{minipage}
	}
	\vfill
	\subfloat{
		\label{fig:subfig_f}
		\begin{minipage}[t]{0.2\textwidth}
			\centering
			\includegraphics[angle=0,width=1\textwidth]{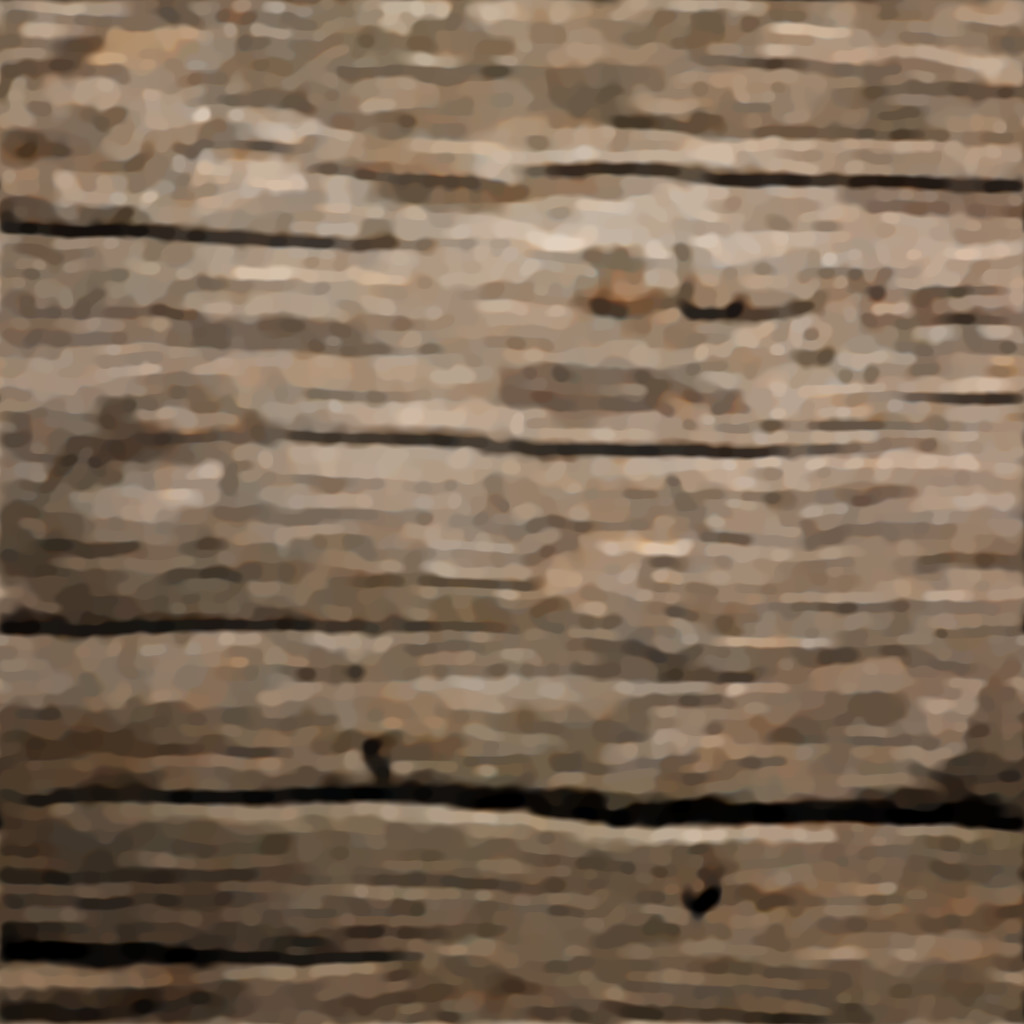}
		\end{minipage}
	}
	\subfloat{
	\label{fig:subfig_f}
	\begin{minipage}[t]{0.2\textwidth}
		\centering
		\includegraphics[angle=0,width=1\textwidth]{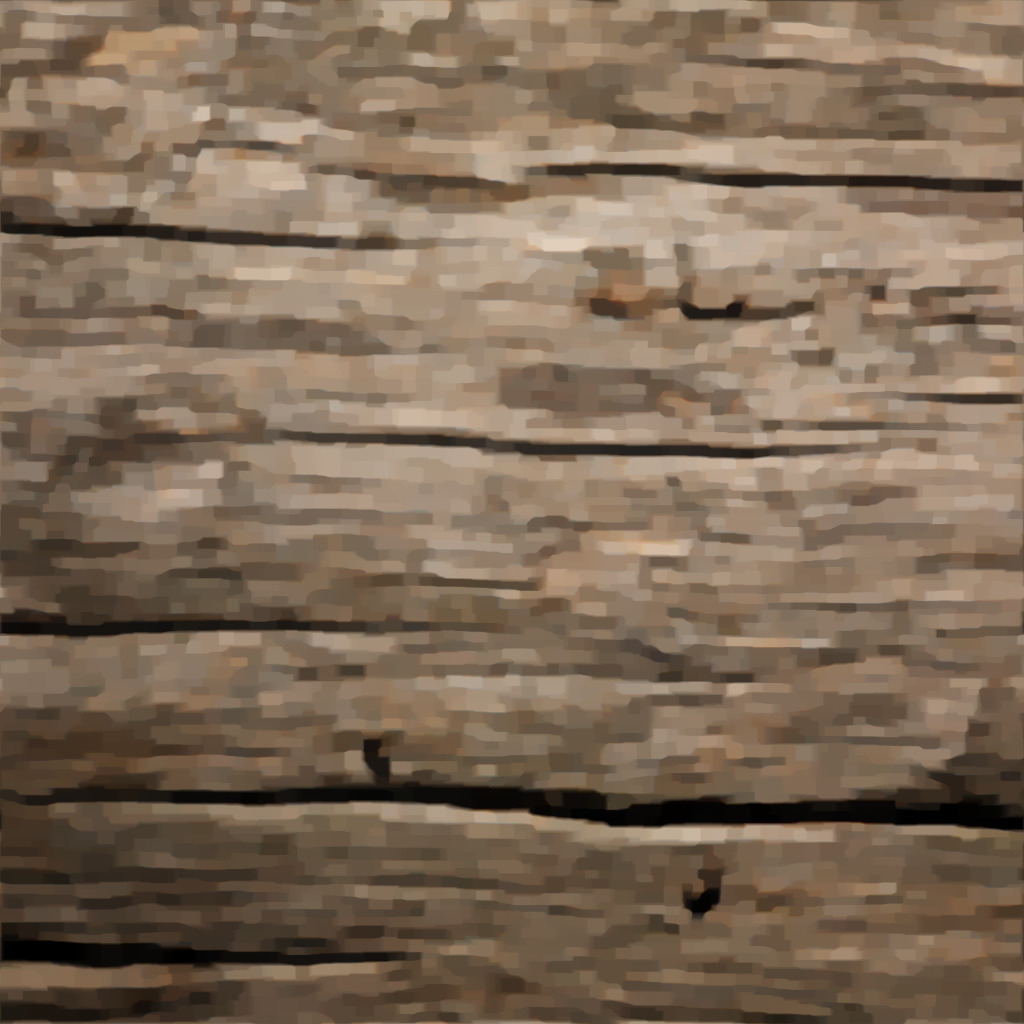}
	\end{minipage}
}
\vfill
	\subfloat{
		\label{fig:subfig_g}
		\begin{minipage}[t]{0.2\textwidth}
			\centering
			\includegraphics[angle=0,width=1\textwidth]{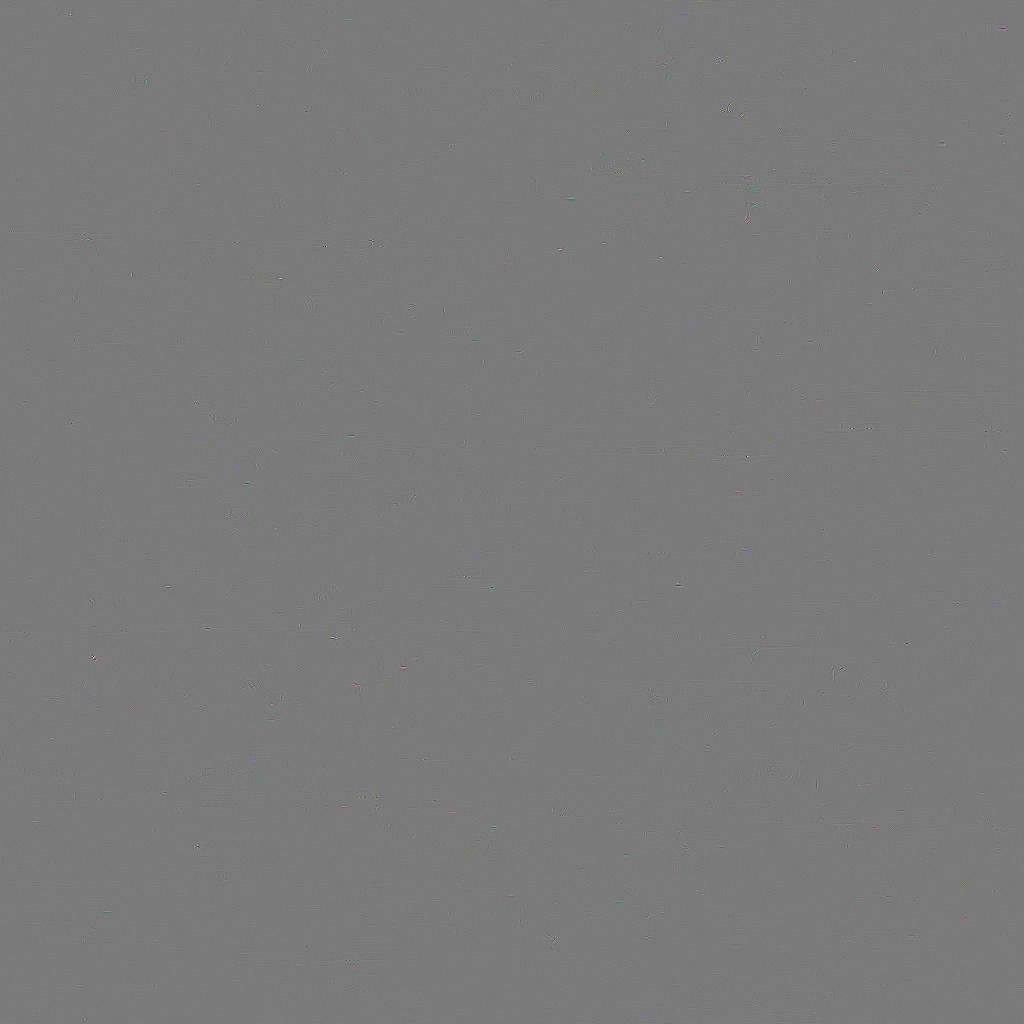}
		\end{minipage}
	}
\subfloat{
	\label{fig:subfig_g}
	\begin{minipage}[t]{0.2\textwidth}
		\centering
		\includegraphics[angle=0,width=1\textwidth]{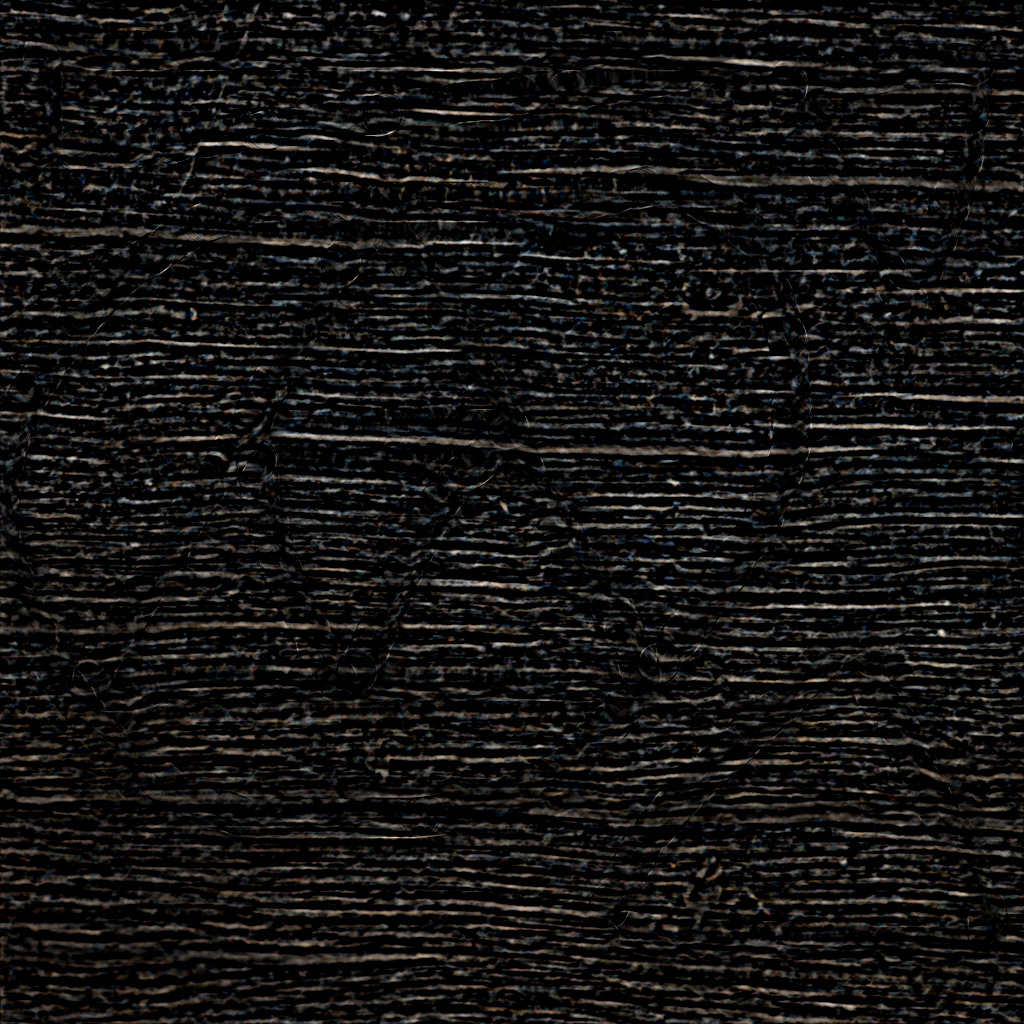}
	\end{minipage}
}
	\vfill
	\subfloat{
		\label{fig:subfig_h}
		\begin{minipage}[t]{0.2\textwidth}
			\centering
			\includegraphics[angle=0,width=1\textwidth]{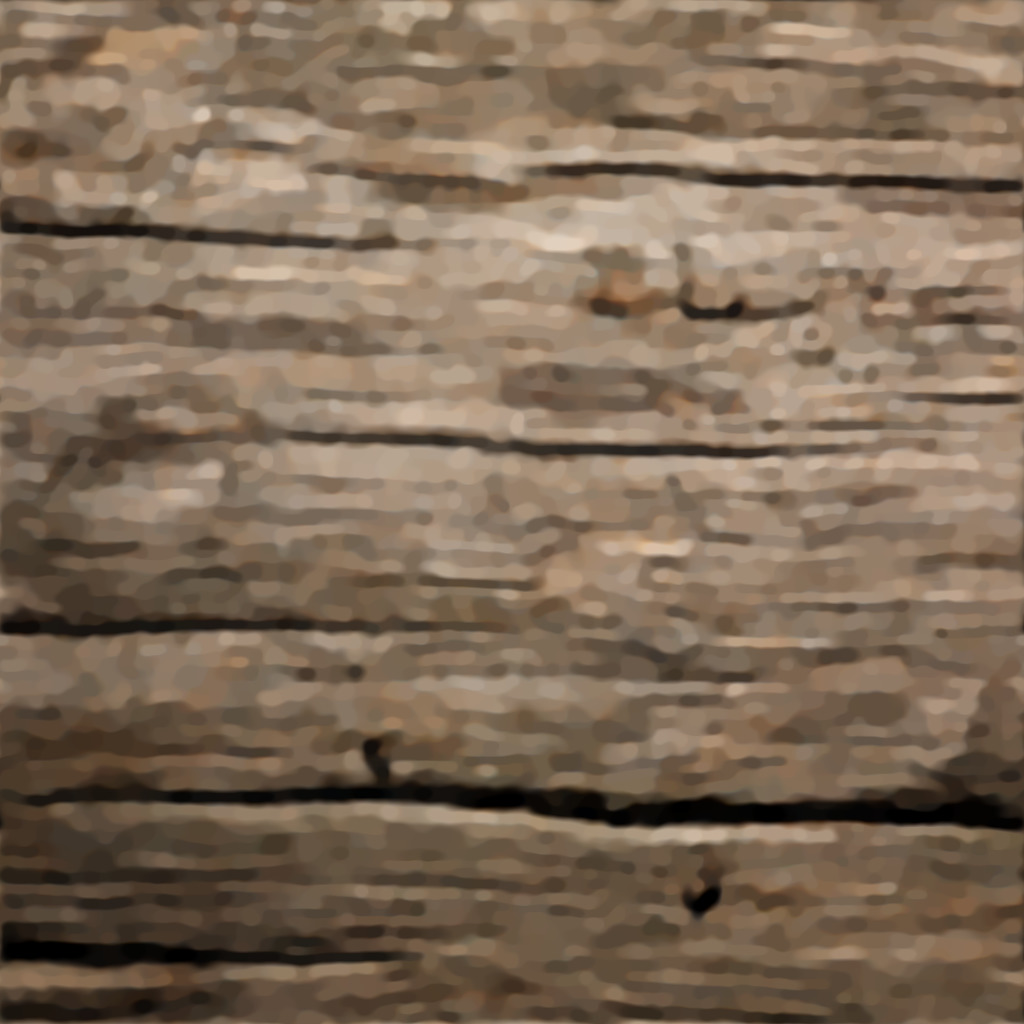}
		\end{minipage}
	}
	\subfloat{
		\label{fig:subfig_h}
		\begin{minipage}[t]{0.2\textwidth}
			\centering
			\includegraphics[angle=0,width=1\textwidth]{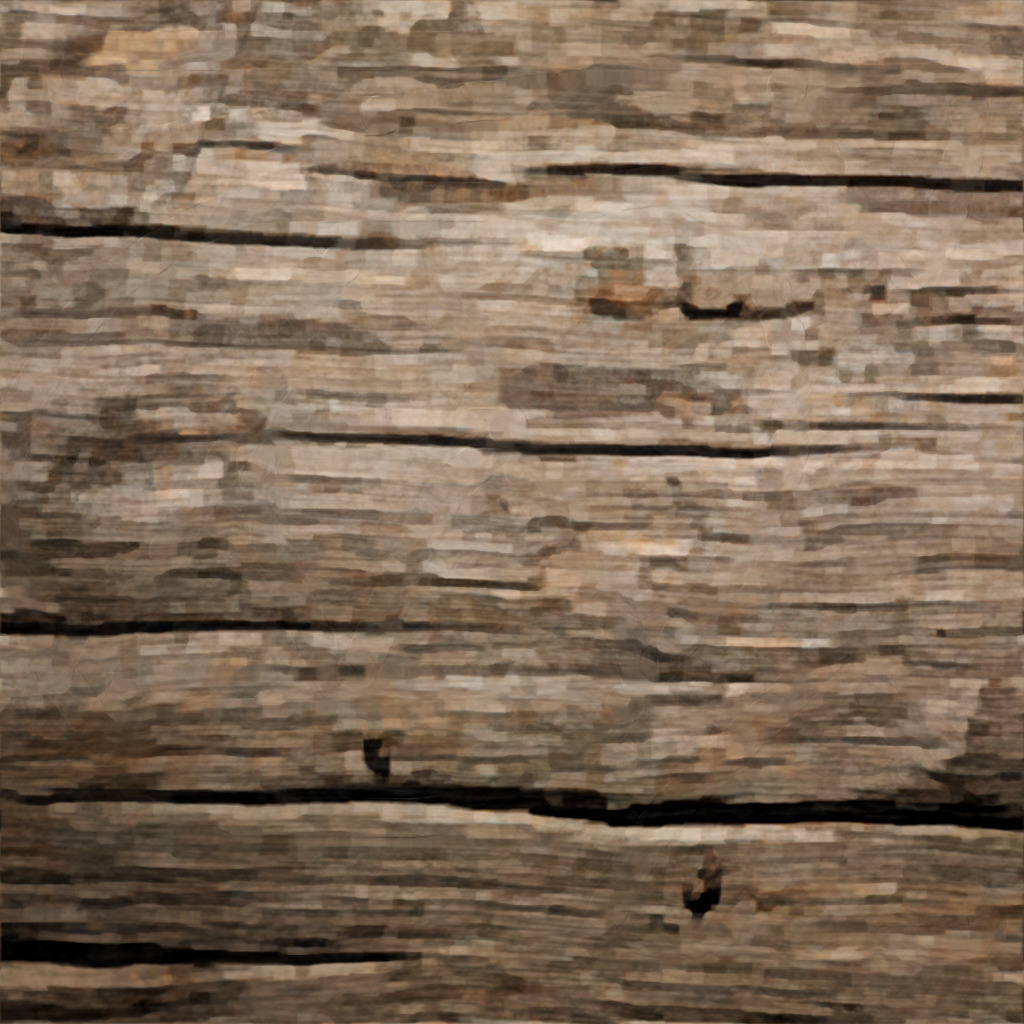}
		\end{minipage}
	}
	\caption{Image decomposition and restoration on the image with blurry and missing pixels (\emph{Wood} image in Fig.~\ref{test_jpg} with blurry (``Out-of-focus(20)'') and missing pixels). From top to bottom are the degraded image, cartoon parts, texture parts and restored images, respectively. Left column: Decomposition by the ADMGB. Right column: Decomposition by the dADMM.}
	\label{blurring_inpainted_jpg}
\end{figure}

We implement two algorithms (the ADMGB and dADMM) in this part. For the ADMGB, we take $\tau = 5\times10^{-5}$, $\mu =1\times10^{-5}$ and $\beta_{1} = \beta_{2} =\beta_{3} = 1\times10^{-2}$ . And for the dADMM, we take $\tau = 5\times10^{-3}$, $\mu =3\times10^{-3}$ and $\sigma=3\times10^{3}$.  All results in this case are displayed in Table~\ref{deblurring_inpainting_ADMGB_ADMMD_table} and Fig.~\ref{blurring_inpainted_jpg}. And the result images (cartoon parts, texture parts and restored images) for \emph{Wood} image with ``Out-of-focus(20)'' plus missing pixels are showed in Fig.~\ref{blurring_inpainted_jpg}.  We report in Table~\ref{deblurring_inpainting_ADMGB_ADMMD_table} the detailed numerical results for the ADMGB and dADMM on the image decomposition and restoration of the images with blurry and missing pixels. It can be observed from Table~\ref{deblurring_inpainting_ADMGB_ADMMD_table} that the PSNR values by the dADMM are higher than those by the ADMGB. Moreover, the dADMM is obviously faster than the ADMGB.

Now we take a further look at the variations of the PSNR values and the KKT residuals with respect to the iterations. Fig.~\ref{psnr_kkt_KSH_jpg} shows the variations of the PSNR values and the KKT residuals with respect to the iterations for the ADMGB and dADMM for (\emph{h}) in Fig.~\ref{test_jpg} with blurry (``Out-of-focus(20)'') and missing pixels, respectively. Firstly, for the ADMGB, the PSNR value can hardly improve after the first several iterations, while for the dADMM, the PSNR value increase all the time though it increases very slow at the later stage. Secondly, the KKT residuals by the dADMM is obviously smaller than that by the ADMGB. Note that the reason lies in the fact that several variables are introduced for the ADMGB which directly lead to a smaller iteration step size than the dADMM.

\begin{figure}[htbp]
	\centering	
	\subfloat{
		\label{fig:subfig_e}
		\begin{minipage}[t]{0.45\textwidth}
			\centering
			\includegraphics[angle=0,width=1\textwidth,height=0.2\textheight]{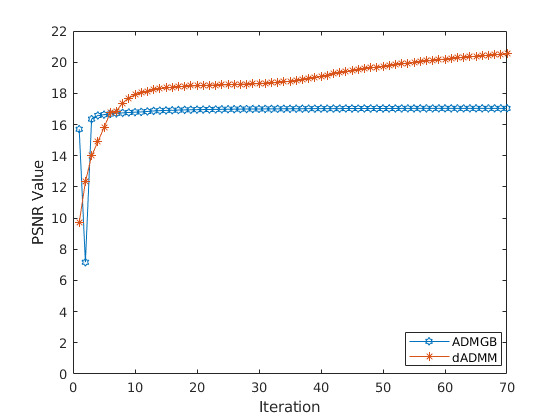}
		\end{minipage}
	}
	\vfill
	\subfloat{
		\label{fig:subfig_f}
		\begin{minipage}[t]{0.45\textwidth}
			\centering
			\includegraphics[angle=0,width=1\textwidth,height=0.2\textheight]{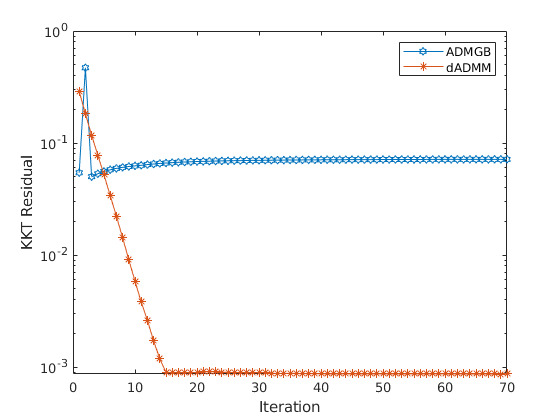}
		\end{minipage}
	}
	\caption{From top to bottom are the variations of the PSNR values and KKT residuals with respect to the iterations for the ADMGB and dADMM for (\emph{h}) in Fig.~\ref{test_jpg} with blurry (``Out-of-focus(20)'') and missing pixels, respectively. }
	\label{psnr_kkt_KSH_jpg}
\end{figure}

\section{Conclusion}
In this paper, we developed the dADMM to solve the image decomposition and restoration problem with blurry and/or missing pixels. The global convergence and the local linear convergence rate of the algorithm were also given. The numerical simulation results also demonstrated that our proposed algorithm is robust and efficient, and can obtain relatively more higher SNRs for various image decomposition and restoration problems. Thus, it shows that even for the same model problem, the choice of the algorithm is also important.

\section*{Appendices}
Note that to compute $\mbox{Prox}_{\sigma p}(\cdot)$ is equivalent to solve the following optimization problem
\begin{eqnarray}
\bar{x}=\arg\min_{x}\,\,\sigma\||\nabla x|\|_{1} + \frac{1}{2}\|x-y\|_{2}^{2}.\label{eq:Proxp}
\end{eqnarray}
Therefore,we adopt the algorithm in \cite{Condat} to solve problem \eqref{eq:Proxp}. For more detailed information, one may see \cite{Condat}.

Furthermore, to compute $\mbox{Prox}_{\sigma q}(\cdot)$ is equivalent to solve the following optimization problem
\begin{eqnarray*}
\bar{x}=\arg\min_{x}\,\,\sigma\||x|\|_{s} + \frac{1}{2}\|x-y\|_{2}^{2},
\end{eqnarray*}
where $\sigma>0$.
\begin{itemize}
	\item  If $s=1$,
	\[
	\bar{x} = y-\max\{\min\{y,\sigma\},-\sigma\}.
	\]
	\item If $s=2$,
	\[
	\bar{x}=y-\min\{\||y|\|,\sigma \}\frac{y}{\||y|\|}.
	\]
	\item If $s=\infty$,
	\[
	\bar{x} = y-\mathcal{P}_{\Omega}(y),
	\]
	where $\mathcal{P}_{\Omega}(\cdot)$ denotes the projection operator onto
	\[
	 \Omega=\{y\,\,|\,\,\||y|\|_{1}\le\sigma \}.
	 \]
\end{itemize}

\section*{Acknowledgements}
We would like to thank Professor Wenxing Zhang at University of Electronic Science and Technology of China for many useful discussions and sharing the code for us.

\begin{IEEEbiography}[{\includegraphics[width=1in,height=1.25in,clip,keepaspectratio]{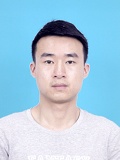}}]{Qingsong Wang}received the Bachelor's degree from Southwest Jiaotong University, Chengdu, China.
He is now a Master student at School of Mathematics, Southwest Jiaotong University. His research interests include numerical optimization, image processing.
\end{IEEEbiography}

\begin{IEEEbiography}[{\includegraphics[width=1in,height=1.25in,clip,keepaspectratio]{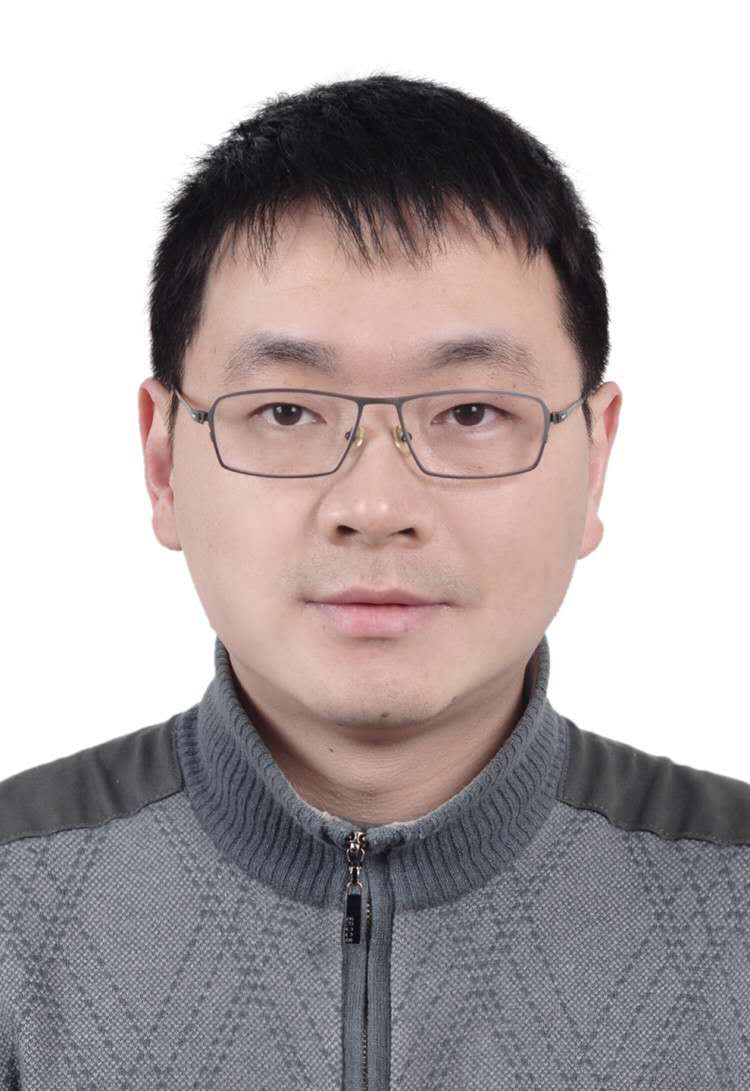}}]{Chengjing Wang}
	received the B.S. and the Ph.D. degrees from Zhejiang University, Hangzhou, China. He is currently an Associate Professor at School of Mathematics, Southwest Jiaotong University, Chengdu, China. His current research interests focus on the theories and algorithms of large-scale optimization and its applications in statistics learning.
\end{IEEEbiography}

\begin{IEEEbiography}[{\includegraphics[width=1in,height=1.25in,clip,keepaspectratio]{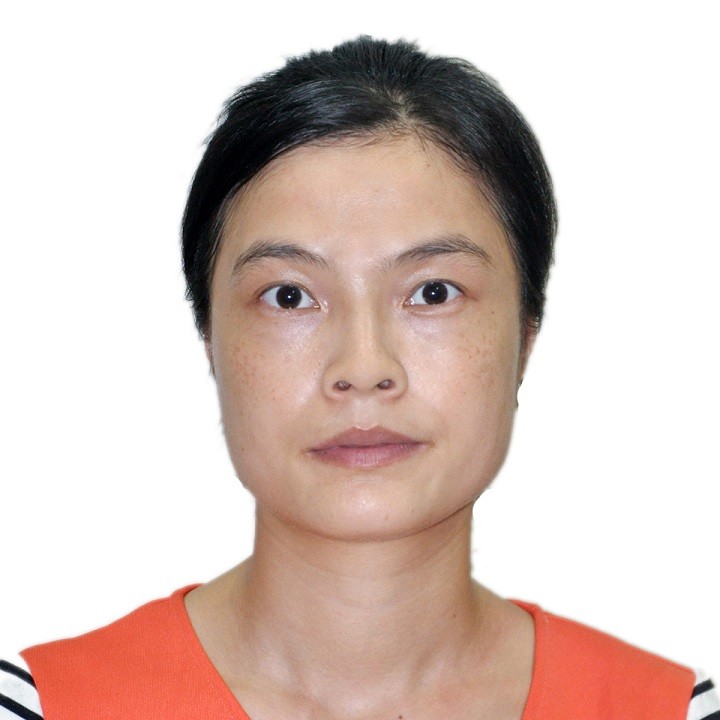}}]{Peipei Tang}
	received the B.S. and the Ph.D. degrees from Zhejiang University, Hangzhou, China.
She is currently an Associate Professor at School of Computer and Computing Science, Zhejiang University City College, Hangzhou, China. Her current research interests include numerical optimization, with its applications in data mining and statistics learning.
\end{IEEEbiography}
 \vspace{.5em}
\begin{IEEEbiography}[{\includegraphics[width=1in,height=1.25in,clip,keepaspectratio]{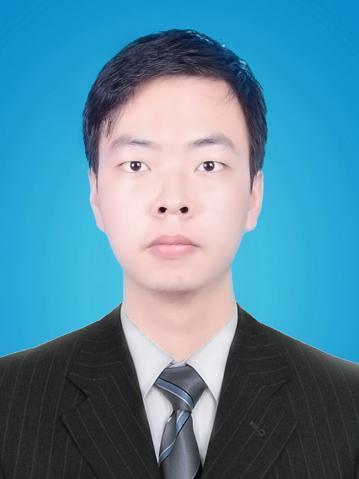}}]{Dunbiao Niu}
	received the B.S. degree in statistics from Southwest Jiaotong University, Chengdu, China in 2016. He is currently pursuing the M.S. degree in probability theory and mathematical statistics at Sichuan University. His current research interests mainly include statistical learning and large-scale optimization.
\end{IEEEbiography}


\begin{thebibliography}{1}
	
\bibitem{Rudin_Osher_Fatemi} L. Rudin, S. Osher, and E. Fatemi, ``Nonlinear total variation based noise removal algorithms,'' Phys. D., vol. 60, no. 1, pp. 259-268, 1992.

\bibitem{Ng_Yuan_Zhang} M.K. Ng, X.M. Yuan, and W.X. Zhang, ``Coupled variational image decomposition and restoration model for blurred cartoon-plus-texture images with missing pixels,'' IEEE Image proc., Vol. 22, no. 6, pp. 2233-2246, 2013.

\bibitem{Aujol_Gilboa_Chan_Osher} J.-F. Aujol, G. Gilboa, T. Chan, and S. Osher, ``Structure-texture image decomposition-modeling, algorithms, and parameter selection,'' Int. J. Comput. Vis., vol. 67, no. 1, pp. 111-136, 2006.

\bibitem{Vese_Osher} L. Vese and S. Osher, ``Modeling textures with total variation minimization and oscillating patterns in image processing,'' J. Sci. Comput., vol. 19, no. 1, pp. 553-572, 2003.

\bibitem{Fazel_Pong_Sun_Tseng} M. Fazel, T.K. Pong, D.F. Sun, and P. Tseng, ``Hankel matrix rank minimization with applications to system identification and realization,'' SIAM J. Matrix Anal. Appl., vol. 34, no. 3, pp. 946-977, 2013.

\bibitem{Han_Sun_Zhang} D.R. Han, D.F. Sun, L.W. Zhang, ``Linear rate convergence of the alternating direction method of multipliers for convex composite programming,'' Math. Oper. Res., Vol 43, no. 2, pp. 622-637, 2017.

\bibitem{Condat} L. Condat, ``A direct algorithm for 1D total variation denoising,'' IEEE Signal Proc. Letters, vol. 20, no. 11, pp. 1054-1057, 2013.

\bibitem{Aujol_Chambolle} J.-F. Aujol and A. Chambolle, ``Dual norms and image decomposition models,'' Int. J. Comput. Vis., vol. 63, no. 1, pp. 85-104, 2005.

\bibitem{Maure_Aujol_Peyre} P. Maure, J.-F. Aujol, and G. Peyr\'{e}, ``Locally parallel texture modeling,'' SIAM J. Imag. Sci., vol. 4, no. 1, pp. 413-447, 2011.

\bibitem{Meyer} Y. Meyer, ``Oscillating patterns in image processing and nonlinear evolution equations,'' (University Lecture Series), vol. 22, Providence, RI, USA: AMS, 2002.

\bibitem{He_Tao_Yuan} B.S. He, M. Tao, and X.M. Yuan, ``Alternating direction method with
Gaussian back substitution for separable convex programming,'' SIAM J. Optim., vol. 22, no. 2, pp. 313-340,  2012.

\bibitem{Cai_Chan_Shen} J.F. Cai, R.H. Chan, and Z. Shen, ``Simultaneous cartoon and texture
inpainting,'' Inverse Prob. Imaging, vol. 4, no. 3, pp. 379-395, 2010.

\bibitem{Glowinski_Marrocco} R. Glowinski and A. Marrocco, ``Sur l'approximation par \'{e}l\'{e}ments finis d'ordre un et la r\'{e}solution par p\'{e}nalisation-dualit\'{e} d'une classe de probl\.{e}mes de Dirichlet non lin\'{e}aires,'' Revue Fr. Autom. Inform. Rech.
Opér., Anal. Num\'{e}r., vol. 2, pp. 41-76, 1975.

\bibitem{Aujol_Aubert} J.-F. Aujol, G. Aubert, L. Blanc-F\'{e}rau, and A. Chambolle, ``Image decomposition into a bounded variation component and an oscillating component,'' J. Math. Imaging Vis., vol. 22, no. 1, pp. 71-88, 2005.

\bibitem{Adams} R. Adams, {\em Sobolev Spaces}, Pure and Applied Mathematics. San Francisco, CA, USA: Academic, 1975.

\bibitem{Vaillo_Caselles_Mazon} F. Andreu-Vaillo, V. Caselles, and J.M. Maz\'{o}n, {\em Parabolic quasilinear equations minimizing linear growth functionals}. Cambridge, MA: Birkhauser, 2004.

\bibitem{Osher_Sole_Vese} S. Osher, A. Sole, and L. Vese, ``Image decomposition and restoration using total variation minimization and the $H^{-1}$ norm,'' Multiscale Model. Simul., vol. 1, no. 3, pp. 349-370, 2003.

\bibitem{Fadili_Starck_Bobin_Moudden} M. Fadili, J. Starck, J. Bobin, Y. Moudden, ``Image decomposition and separation using sparse representations: an overview,'' Proc.
IEEE,  vol. 98, no. 6, pp. 983-994,2010.

\bibitem{Yin_Goldfarb_Osher} W. Yin, D. Goldfarb, S. Osher, ``A comparison of three total variation based texture extraction models,'' J. Vis. Commun. Image R., vol. 18, no. 3, pp. 240-252, 2007.

\bibitem{Rockafellar} T. Rockafellar, {\em Convex Analysis}. Princeton, NJ, USA: Princeton University Press, 1970.

\bibitem{Nocedal_Wright} J. Nocedal, and S.J. Wright, {\em Numerical Optimization}. second edition, New York: Springer, 2006.

\bibitem{Kuhn_Tucker}  H.W. Kuhn,  A.W. Tucker, ``Nonlinear programming,'' Proceedings of 2nd Berkeley Symposium. Berkeley: University of California Press., pp. 481-492, 1951.
\end{thebibliography}
\end{document}